\documentclass[12pt, a4paper]{article}
\usepackage{times,amsmath,epsfig,graphicx,balance,multirow,amssymb,geometry,bm,diagbox,amsthm,xspace,xcolor}
\usepackage{enumerate}
\usepackage{subcaption}
\usepackage{booktabs}
\usepackage[ruled, vlined, linesnumbered]{algorithm2e}
\usepackage{authblk}

\providecommand{\eg}{\emph{e.g.},\xspace}
\newcommand{\ie}{\emph{i.e.},\xspace}

\newcommand{\wrt}{\emph{w.r.t.},\xspace}

\newcommand\figref[1]{Figure~\ref{#1}}

\newcommand\tabref[1]{Table~\ref{#1}}

\newcommand\secref[1]{Section~\ref{#1}}
\newcommand\equref[1]{Equation~(\ref{#1})}

\newcommand\algoref[1]{Algo.~\ref{#1}}
\newcommand\lineref[1]{Line~\ref{#1}}
\newcommand\lemref[1]{Lemma~\ref{#1}}
\newcommand\theref[1]{Theorem~\ref{#1}}

\newcommand\expref[1]{Example~\ref{#1}}
\newcommand\defref[1]{Definition~\ref{#1}}

\newcommand{\fakeparagraph}[1]{\vspace{1mm}\noindent\textbf{#1.}}
\newtheorem{definition}{Definition}     %[section]
\newtheorem{example}{Example}     %[section]
\newtheorem{lemma}{Lemma}     %[section]
\newtheorem{theorem}{Theorem}     %[section]
\ifodd 0

\else

\fi

\begin{document}

\title{Accurate and Efficient Trajectory-based Contact
Tracing with Secure Computation and
Geo-Indistinguishability}

\author[1]{Maocheng Li}
\author[1]{Yuxiang Zeng}
\author[2]{Libin Zheng}
\author[1,4]{Lei Chen}
\author[3]{Qing Li}
\affil[1]{The Hong Kong University of Science and Technology, Hong Kong, China
\{csmichael,yzengal,leichen\}@cse.ust.hk}

\affil[2]{Sun Yat-sen University, Guangzhou, China zhenglb6@mail.sysu.edu.cn}
\affil[3]{The Hong Kong Polytechnic University, Hong Kong, China
csqli@comp.polyu.edu.hk}
\affil[4]{The Hong Kong University of Science and Technology (Guangzhou), Guangzhou,
China}

\date{}
\maketitle

\begin{abstract}

  Contact tracing has been considered as an effective measure to limit the transmission of infectious disease such as COVID-19. Trajectory-based contact tracing compares the trajectories of users with the patients, and allows the tracing of both direct contacts and indirect contacts. Although trajectory data is widely considered as sensitive and personal data, there is limited research on how to securely compare trajectories of users and patients to conduct contact tracing with excellent accuracy, high efficiency, and strong privacy guarantee. Traditional Secure Multiparty Computation (MPC) techniques suffer from prohibitive running time, which prevents their adoption in large cities with millions of users. In this work, we propose a technical framework called ContactGuard to achieve accurate, efficient, and privacy-preserving trajectory-based contact tracing. It improves the efficiency of the MPC-based baseline by selecting only a small subset of locations of users to compare against the locations of the patients, with the assist of Geo-Indistinguishability, a differential privacy notion for Location-based services (LBS) systems. Extensive experiments demonstrate that ContactGuard runs up to 2.6$\times$ faster than the MPC baseline, with no sacrifice in terms of the accuracy of contact tracing. 
  
  \textbf{Keywords:} Contact tracing, Differential privacy, Spatial database.
\footnotetext[1]{Copyright may be transferred without notice, after which this version no longer be accesiible.}
\end{abstract}
\section{Introduction}
\label{sec:intro}

Contact tracing has been considered as one of the key epidemic control measures to limit the transmission of infectious disease such as COVID-19 
and Ebola \cite{lewis2020many}. In contrast to traditional contact tracing conducted normally by interviews and manual tracking, digital contact tracing %\cite{cebrian2021past} 
nowadays rely on mobile devices to track the visited locations of users, and are considered as more accurate, efficient, and scalable \cite{rodriguez2021population,salathe2020early}.
Taking the contact tracing of COVID-19 as an example, more than 46 countries/regions have launched contact tracing applications \cite{lewis2020many}, \eg TraceTogether in Singapore and LeaveHomeSafe in Hong Kong SAR. %\footnote{\url{www.leavehomesafe.gov.hk}}. 

As a complementary technique to the Bluetooth-based contact tracing applications such as the Exposure Notification developed by Apple and Google\footnote[2]{https://covid19.apple.com/contacttracing}, \textit{trajectory-based contact tracing} \cite{kato2020secure,DBLP:conf/icde/DaA0S21} compares the trajectories of users with the patients, and allows the tracing of both the direct contacts and the indirect contacts. The direct contacts happen when the users and the patients co-visit the same location at the \textit{same} time. In contrast, the indirect contacts normally happen when users and patients visit the same location at \textit{different} times, and the location (\eg environment, surface of objects) is contaminated and becomes the transmission medium of the virus. Because its capability of tracing both direct and indirect contacts, trajectory-based contact tracing has been widely deployed in countries/regions where there is a stricter policy of COVID-19 control.  For example, in China, prior to entering high-risk locations such as restaurants, bars, and gyms, users need to check-in with a health code, which is linked with their personal identities. 

Despite the usefulness of trajectory-based contact tracing, the privacy of trajectory data is an obvious issue. In fact, privacy concerns prevent the contact tracing application from being widely adopted in cities like Hong Kong, and there are extreme cases that citizens use two different mobile phones, one for their daily use, and the other one solely for the purpose of fulfilling the check-in requirement when entering the high-risk locations, in order to prevent any privacy breach. Existing privacy-preserving trajectory-based contact tracing research relies on specific hardware (Trusted Hardware such as Intel SGX \cite{kato2020secure}), but the security guarantee heavily depends on the hardware, and it leads to poor portability. In fact, Intel plans to stop its support for SGX from its 11th and 12th generation processors. %\cite{pcgamer2022sgx}
Thus, in this work, we aim to develop a hardware-independent software-based solution, which can provide high accuracy of contact tracing, excellent execution efficiency, and strong privacy guarantee.

\begin{figure}[t]
	\centering \vspace{0ex}
	   \begin{subfigure}[b]{0.49\textwidth}
		   \centering
		   \includegraphics[width=\textwidth]{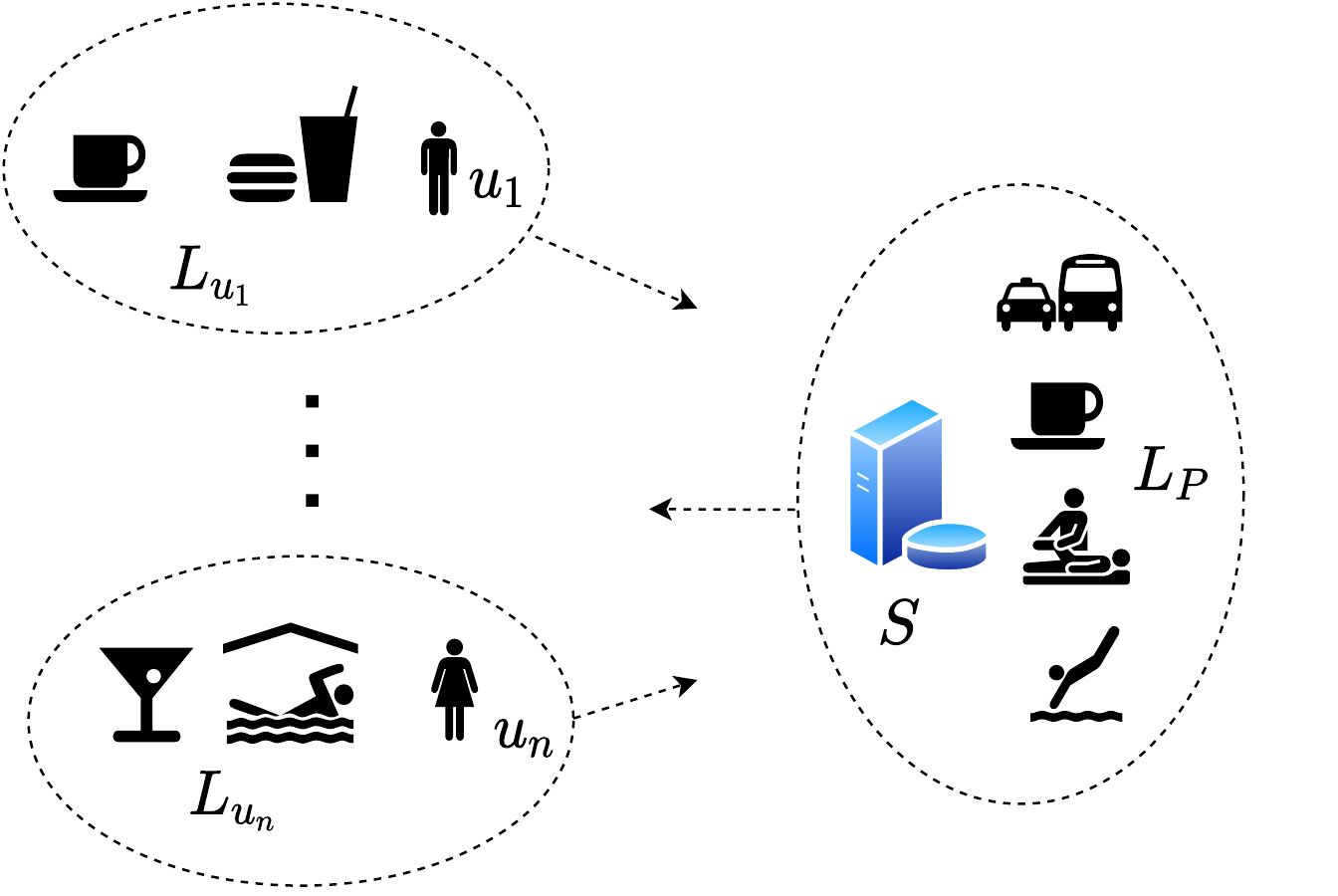}
		   \caption{\small A simplified example of our problem.}
		   \label{subfig:problem}
	   \end{subfigure}
	   \hfill
		\begin{subfigure}[b]{0.49\textwidth}
		   \centering
		   \includegraphics[width=\textwidth]{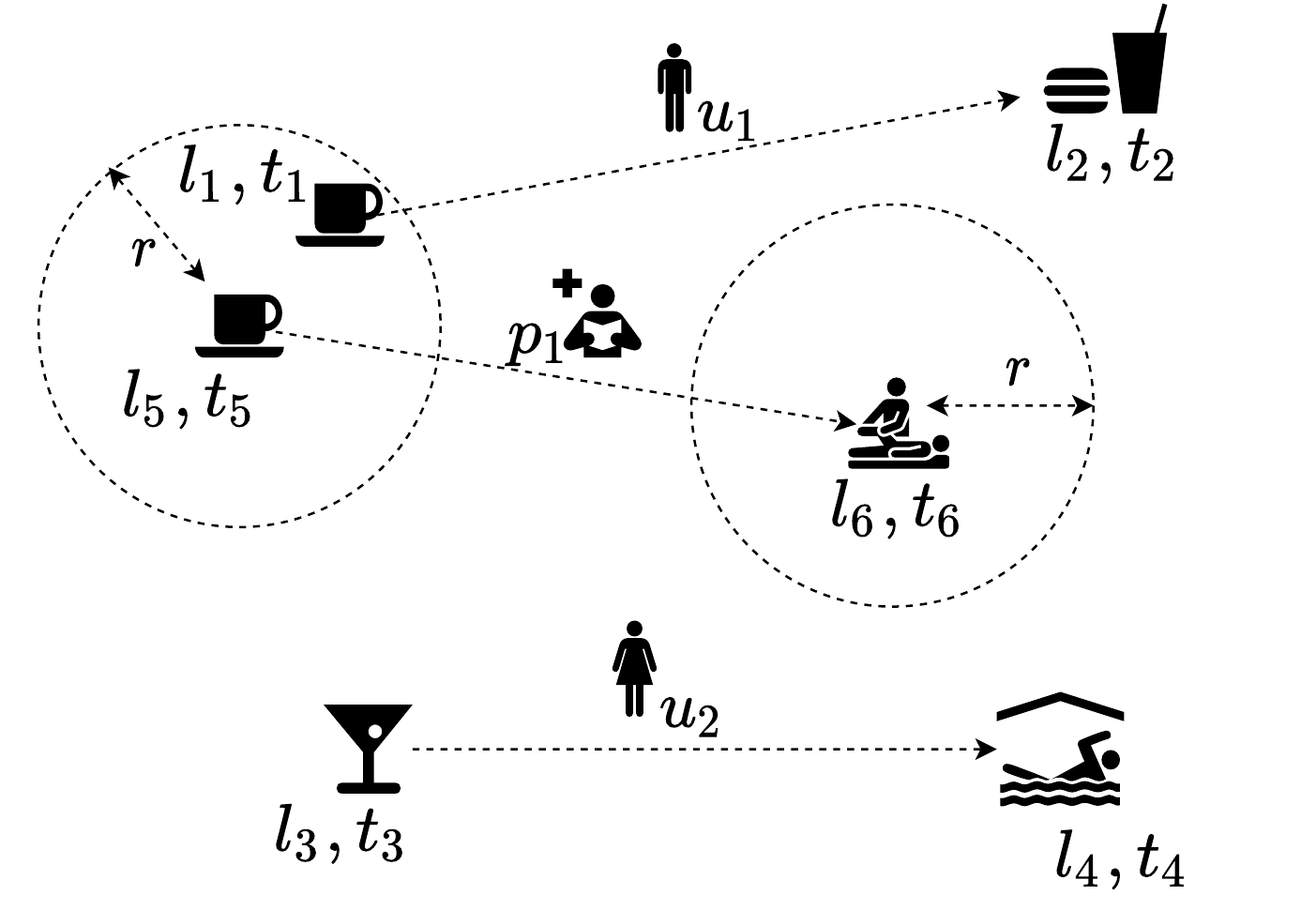}
		   \caption{An example of the close-contact.}
		   \label{subfig:example_close_contact}
	   \end{subfigure}
	\caption{\small \small (a) A simplified example of our problem. $L_{u_i}$ denotes the visited location of a user $u_i$. $L_P$ denotes the collection of all visited locations of all patients. (b) An example of the close-contact. $u_1$ is a close contact, while $u_2$ is not. }	
	\label{fig:problem_and_example}
   \end{figure}

In this paper, we formulate the \textit{\underline{P}rivacy-\underline{P}reserving \underline{C}ontact \underline{T}racing (PPCT)} problem, which is illustrated with the simplified example in \figref{subfig:problem}. We aim to correctly identify whether each user $u_i$ is a close contact or not. If a user $u_i$ co-visited some location with some patient under some spatial and temporal constraints (\eg co-visit the same location with 5 meters apart, and within a time window of 2 days), then our proposed solution should correctly identify the user $u_i$ as a close contact. The key challenge of PPCT is that there is a strong privacy requirement: during the entire process of contact tracing, the private information (the visited locations of users, together with the timestamps) is strictly protected from the access of other parties. 

The challenge of PPCT goes beyond the strong privacy requirement. On one hand, the contact tracing application requires a high level of accuracy. Taking COVID-19 as an example, some countries spend tremendous amount of resources to identify potential close contacts of patients, and sometimes a whole region of citizens need to go through mandated testing after only one case of patient was found in the area. Thus, our proposed solution needs to avoid false negatives as far as possible, maintaining a reasonably high recall, if not close to 100\%. On the other hand, we aim to offer an efficient solution to the PPCT problem, as traditional Secure Multiparty Computation (MPC) techniques usually induce unpractical running time. As our experimental results show, a na\"ive MPC solution requires more than 5 days  for 1 million users, which is the population of a modern city. Such long running time prohibits its daily execution, because it could not even terminate within 24 hours. 

To tackle the aforementioned challenges, we propose a novel technical framework ContactGuard. ContactGuard improves the efficiency of MPC operations by selecting only small subsets of users' locations to compare with the patients. The subset selection process is assisted by an efficient privacy-preserving mechanism -- Geo-Indistinguishability (Geo-I) for Location-based services (LBS) systems. The basic idea is that each user perturbs his/her true visited locations with Geo-I and submits only the perturbed locations to the server, where the patients' locations are stored. Using the perturbed locations of the user, the server compares them with the patients' locations, and then informs the user about which locations are more ``risky'', because they are closer to the patients. After that, the user could use MPC only on the high-risk locations, to largely reduce the computational overhead.

To summarize, we make the following contributions in this paper:
 
\begin{itemize}
	\setlength\itemsep{0.1em}

	\item[\textbullet] We formally define an important problem, the Privacy-Preserving Contact Tracing (PPCT) problem in \secref{sec:problem}, which addresses the privacy issue in trajectory-based contact tracing. 
	
	\item[\textbullet] We propose a novel solution ContactGuard in \secref{sec:contactGuard}. It combines the strengths of two different privacy-preserving paradigms: differential privacy and secure multiparty computation. The running time is improved significantly because it selects only a small subset of locations to compare with the patients during the MPC operation. 
	
	\item[\textbullet] We conduct extensive experiments to validate the effectiveness and efficiency of ContactGuard in \secref{sec:experiment}. It runs up to 2.6$\times$ faster than the MPC baseline, with no sacrifice in terms of the accuracy of contact tracing. For moderate privacy budget settings for each user, ContactGuard obtains close to 100\% recall and 100\% precision.  

\end{itemize}
In addition, we review related works in \secref{sec:relatedWork} and conclude in \secref{sec:conclusion}. 
\section{Problem Definition}
\label{sec:problem}
In this section, we introduce some basic concepts, the adversary model, and a formal definition of the Privacy-Preserving Contact Tracing (PPCT) problem.

\subsection{Basic Concepts}
\label{subsec:basic_concept}

\begin{definition}[Location]\label{def:location}
A location $l=(x,y)$ represents a 2-dimensional spatial point with the coordinates $(x,y)$ on an Euclidean space. 
\end{definition}

\begin{definition}[Trajectory]\label{def:visited_locations}
	A trajectory is a set of (location, timestamp) tuples indicating the visited location and time that the location is visited. A trajectory $L=\{(l_1,t_1), \ldots, (l_{|L|}, t_{|L|})\}$, where $|L|$ is the number of visited locations in the trajectory. 
	\end{definition}

Examples: let location $l_1 = (300, 500)$ and $l_2$ be another location $l_2=(200, 200)$. An example of trajectory $L_{example} = \{(l_1, $ 2020-06-10 12:28:46)$, (l_2, $ 2020-06-10 17:03:24$)\}$ indicates that location $l_1$ is visited at timestamp 2021-06-10 12:28:46 and location $l_2$ is visited at timestamp 2021-06-10 17:03:24. The size of $L_{example}$ is 2, \ie $|L_{example}|=2$. 

Here, the definition of \textit{trajectory} %\cite{mamoulis2011spatial}
includes both short sequences (\eg a user leaves home in the morning, moves to the working place, and comes back home at night) and longer sequences across multiple days. In our problem setting, in order to trace contacts, the trajectory of a user contains the union of all visited locations of a particular user in the incubation period of the disease (\eg past 14 days for COVID-19 \cite{doi:10.1056/NEJMoa2002032,ferretti2020quantifying}).

\begin{definition}[Users]\label{def:test_user}
	A set $U$ denotes all $n$ users. Each user $u \in U$ is associated with a trajectory $L_u=\{(l_1,t_1), \ldots, (l_{|L_u|}, t_{|L_u|})\}$, following \defref{def:visited_locations},  indicating the locations the user $u$ visited and the time of the visits. 
\end{definition}

Each user needs to be checked against the patients to see whether the user is a contact of some patient or not. The definition of a patient is similar to the one of the user, and we define them separately for easier illustration in the contact tracing problem. 

\begin{definition}[Patients]\label{def:patient}
	A set $P$ denotes all $m$ patients. Each patient $p\in P$ has a trajectory $L_p = \{(l_1,t_1), \ldots, (l_{|L_p|}, t_{|L_p|})\}$, following \defref{def:visited_locations},  indicating the locations the patient $p$ visited and the time of the visits. The set $L_P = L_{p_1} \cup L_{p_2} \cup \ldots \cup L_{p_m}$ is a union of trajectories of all patients $p_1,\ldots,p_m$.
\end{definition}

In our problem setting, we use $L_P$ to denote the aggregated set of all trajectories of patients. It is an important notion because a user $u \in U$ is defined as a \textit{contact} (see a more formal definition in \defref{def:close_contact}) if the user $u$'s trajectory $L_u$ overlaps with some location in $L_P$. It is not necessary to identify which specific patient $p\in P$  leads to the contact. Indeed, as we will explain further in our adversary model and system settings in \secref{subsec:adversary_model}, it suffices to store the union of all trajectories $L_P$ of the patients, without distinguishing each patient's individual trajectories, and it also enhances privacy protection.  

\begin{definition}[Contacts]\label{def:close_contact}
	Given a distance threshold $r$, and a time difference threshold $\delta$, a user $u \in U$ is called the \textbf{contact} if
	the user $u$ has visited a location $l_u$ that is within $r$ distance to some visited location $l_p$ of some patient $p\in P$, and the time difference between the two visits is within $\delta$, i.e.,
	\begin{equation}
		\exists (l_u, t_u) \in L_u \; \exists (l_p, t_p) \in L_P \; ((d_s(l_u, l_p) \leq r) \land (d_t(t_p, t_u) \leq \delta))
	\end{equation}
	where the function $d_s(l_u,l_p)$ represents the \textit{spatial distance} -- the Euclidean distance between locations $l_u$ and $l_p$. $d_t(t_u, t_p)$ represents the \textit{temporal difference} -- the time difference in seconds from an earlier timestamp $t_p$ to a later timestamp $t_u$. 
\end{definition}

We give a toy example in \expref{example:close-contact} to better illustrate the concept of contacts.  

\begin{example} \label{example:close-contact}
As shown in \figref{subfig:example_close_contact}, there is one patient $p_1$ and two users $u_1$-$u_2$. User $u_1$ has a trajectory $L_{u_1} = \{(l_1, t_1), (l_2, t_2)\}$. User $u_2$ has a trajectory $L_{u_2} = \{(l_3, t_3), (l_4, t_4)\}$. The set of patients $P$ (only contains one single patient $p$) has the trajectory $L_P=\{(l_5, t_5), (l_6, t_6)\}$. The locations are shown on the figure, and let us further specify the time. Let $t_1 = $ 2021-06-10 11:00:00 and $t_5 = $ 2021-06-10 10:00:00, indicating that patient $p$ visits the location $l_5$ an hour earlier than the time when user $u$ visits $l_1$. 

Based on \defref{def:close_contact}, let $r$ be the distance shown in the figure, and the time difference threshold be $\delta=$ 2 hours, then the user $u_1$ is a contact, since $d_s(l_1, l_5) \leq r$, \ie $u_1$ has visited $l_1$ and $l_1$ is within distance $r$ to patient $p_1$'s visited location $l_5$, and $d_t(t_1, t_5) = $ 1 hour $ \leq \delta$, \ie the time difference between the two visits is within $\delta$ ($\delta=$ 2 hours). 

On the other hand, user $u_2$ is not a close-contact, because her visited locations, $l_3$ and $l_4$, are not in proximity of any patient's visited locations. 
\end{example}

\subsection{Adversary Model}
\label{subsec:adversary_model}

There are two major roles in our application: the users (the clients) and the government (the server). All trajectories by the patients (represented as $L_P$ as in \defref{def:patient}) are aggregated and stored at the server. The setting follows the real-world situation: governments often collect whereabouts of confirmed patients in order to minimize any further transmission, usually starting with tracing the close contacts of the patients. For example, this is enforced by legislation in Hong Kong\footnote[3]{https://www.elegislation.gov.hk/hk/cap599D!en?INDEX\_CS=N}. 

We adopt a semi-honest model as the adversary model in our problem. Adversary could exist on both the client (the user) and the server side (the government). We assume both the client and the server are curious about other parties' private information, but they are not malicious and follow our designed system protocols. The semi-honest model is a commonly adopted setting in recent privacy-preserving LBS related applications \cite{TongJoS17,DBLP:conf/icde/TaoTZSC020}. 

\subsection{Privacy-Preserving Contact Tracing Problem}
\label{subsec:problem}

Based on the concepts and adversary model introduced previously, we now define the Privacy-Preserving Contact Tracing (PPCT) problem as follows.

\begin{definition}[Privacy-Preserving Contact Tracing (PPCT) problem] \label{def:ppct}
	Given a set $U$ of $n$ users , the trajectory $L_u$  for each user $u\in U$, a set $P$ of $m$ patients, a set $L_P$ containing the union of trajectories for all patients $p_1, \ldots, p_m \in P$, a distance threshold $r$, and a time difference threshold $\delta$,
	the PPCT problem needs to correctly identify each user as a contact (see \defref{def:close_contact}) or not. 

	In addition, PPCT has the following privacy requirements:
	\begin{itemize}
        \setlength\itemsep{0.1em}
		\item[\textbullet] The computation process needs to be differentially private / confidential \wrt each user's trajectory $L_u$. 
		\item[\textbullet] The computation process needs to be differentially private / confidential \wrt $L_P$, the union of trajectories of the patients. 
	\end{itemize}
\end{definition}
\section{Methodology}
%\section{Our Solution ContactGuard}
\label{sec:contactGuard}

In this section, we first introduce two baselines, one based on Secure Multiparty Computation (MPC) and the other one based on Geo-Indistinguishability (Geo-I). %, an extension of differential privacy into the location-based services domain). 
Then, we introduce our proposed method -- ContactGuard.

\subsection{Baselines}

MPC and Geo-I are two different paradigms to satisfy the privacy requirements in the PPCT problem, but they differ greatly in terms of efficiency and accuracy. MPC allows the server and the client to use cryptographic primitives to securely compare each visited location without revealing the exact locations of one party to the other parties. However, it induces significant computation overhead. 

Geo-I \cite{andres13} is an extended notion of differential privacy into the spatial domain (strictly speaking, Geo-I is a variant of local differential privacy notion \cite{duchi2013local,cormode2019answering}). 
It is an efficient mechanism to be applied to the location inputs, however Geo-I injects noise into the inputs (protecting the input location by perturbing it to an obfuscated location). Using the perturbed locations to compare trajectories would then lead to errors and reduce the accuracy of contact tracing.

\subsubsection{MPC baseline}
\label{subsec:mpc_baseline}

The MPC baseline is to directly apply existing secure multiparty computation techniques to our problem. In our setting, each time, we treat a client (a user) and the server (which holds the patients' data) as two parties. Each of the party holds their own visited locations, which are private. Then, they use MPC operations to compare their visited locations and check whether the user is a contact or not, according to \defref{def:close_contact}. Note that the baseline computes the \textit{exact} result, which means that it does not lose any accuracy.

\subsubsection{Geo-I baseline}
\label{subsec:geo-i_baseline}

In Geo-I baseline, each user perturbs his/her trajectory $L_u$ to a protected noisy location set $L'_u$ using Geo-I (\algoref{algo:perturb_location_set}). Then, the user submits $L'_u$ to the server. The server directly compares $L'_u$ with the patients' locations $L_P$. If there exists some perturbed location $l' \in L'_u$ that is within distance $r'$ (a system parameter), then $u$ is identified as a contact.

    \begin{algorithm}[H]
      \DontPrintSemicolon
    	\KwIn{$\epsilon, L_u$. }
    	\KwOut{$L'_u$.}
    
        $L'_u := \{\}$ \;
        $\epsilon'= \epsilon / |L_u|$ \label{line:eps_divide}\;
    	\ForEach{$l \in L_u$}{
            $l' = $ Geo-I($\epsilon', l$) \label{line:geo_i}\;
            $L'_u$.insert($l'$) \;
    	}
    
    	\Return{$L'_u$}\;
    	\caption{\texttt{Perturb\_Location\_Set} }
	    \label{algo:perturb_location_set}
    \end{algorithm}

\algoref{algo:perturb_location_set} is a generic method to obtain a set of perturbed locations $L'_u$, given a set of original locations $L_u$ and a total privacy budget of $\epsilon$. At \lineref{line:eps_divide}, the total privacy budget $\epsilon$ is equally divided into $|L_u|$ shares. Thus, each share is $\epsilon'=\epsilon/|L_u|$. Then we perturb each location $l \in L_u$ into a new location $l'$ by the Geo-I mechanism. 
All perturbed locations $l'$ compose the set $L'_u$ of the perturbed locations.

Note that in the original definition for the trajectory $L_u$ in \defref{def:visited_locations}, each visited location $l$ is associated with a timestamp that the location is visited (see Example~\ref{example:close-contact}). We omit the timestamp information for each visited location, and perturb \textit{only} the locations to perturbed locations. We abuse the notation $L_u$ and $L'_u$ to denote only the spatial locations (\defref{def:location}) and the generated perturbed locations. 

\fakeparagraph{Time Complexity	} For \algoref{algo:perturb_location_set} (the client side), the time complexity is $\mathcal{O}(|L_u|)$, linear to the number of visited locations given in the input $L_u$. For the server side, because it compares each location of $L'_u$ against every location of $L_P$, the time complexity is $\mathcal{O}(|L'_u||L_P|) \to \mathcal{O}(|L_u||L_P|)$. For the server $S$ to finish the processing of all users, the total time complexity is thus $\mathcal{O}(\sum_u |L_u||L_P|) \to \mathcal{O}(|U||L_P|\max_u|L_u|)$. 

\fakeparagraph{Privacy Analysis} We defer the privacy analysis of \algoref{algo:perturb_location_set} to \secref{subsec:contact_guard}, as it is the same as the privacy analysis of ContactGuard.

\subsection{Our Solution ContactGuard}
\label{subsec:contact_guard}

In this section, we propose a novel method called ContactGuard. As we have previously introduced, the MPC baseline provides excellent accuracy, because it invokes computationally heavy MPC operations to compare all visited locations of users against the visited locations of the patients. The weakness of the MPC baseline is its poor efficiency. On the other hand, the Geo-I baseline is efficient because each client (user) applies an efficient Geo-I mechanism to change his/her visited locations to perturbed ones. But the inaccurate perturbed locations could lead to false positives/negatives for the PPCT problem. 

The ContactGuard aims to combine the advantages of both baselines. On one hand, it uses the exact locations and MPC to determine whether a user is a contact or not to ensure excellent accuracy while providing strong privacy protection. On the other hand, it accelerates the MPC operation by selecting only a subset of users' visited locations to invoke the heavy MPC operations. The selection is done with the help of Geo-I. The overall workflow of ContactGuard is illustrated in \figref{subfig:workflow}. The details are explained next.

\begin{figure}[t]
	\centering \vspace{0ex}
	   \begin{subfigure}[b]{0.49\textwidth}
		   \centering
		   \includegraphics[width=\textwidth]{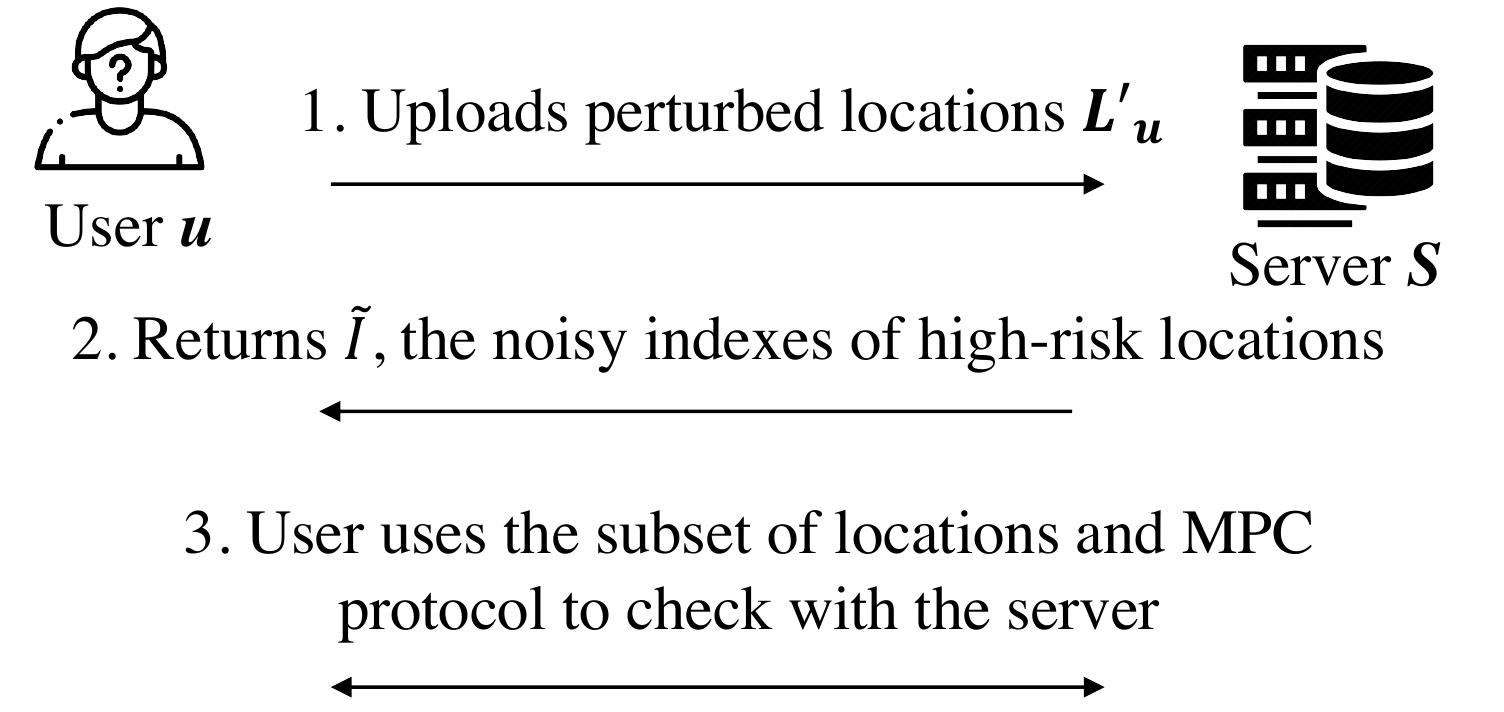}
		   \caption{\small Overall workflow. }
		   \label{subfig:workflow}
	   \end{subfigure}
	   \hfill
		\begin{subfigure}[b]{0.40\textwidth}
		   \centering
		   \includegraphics[width=\textwidth]{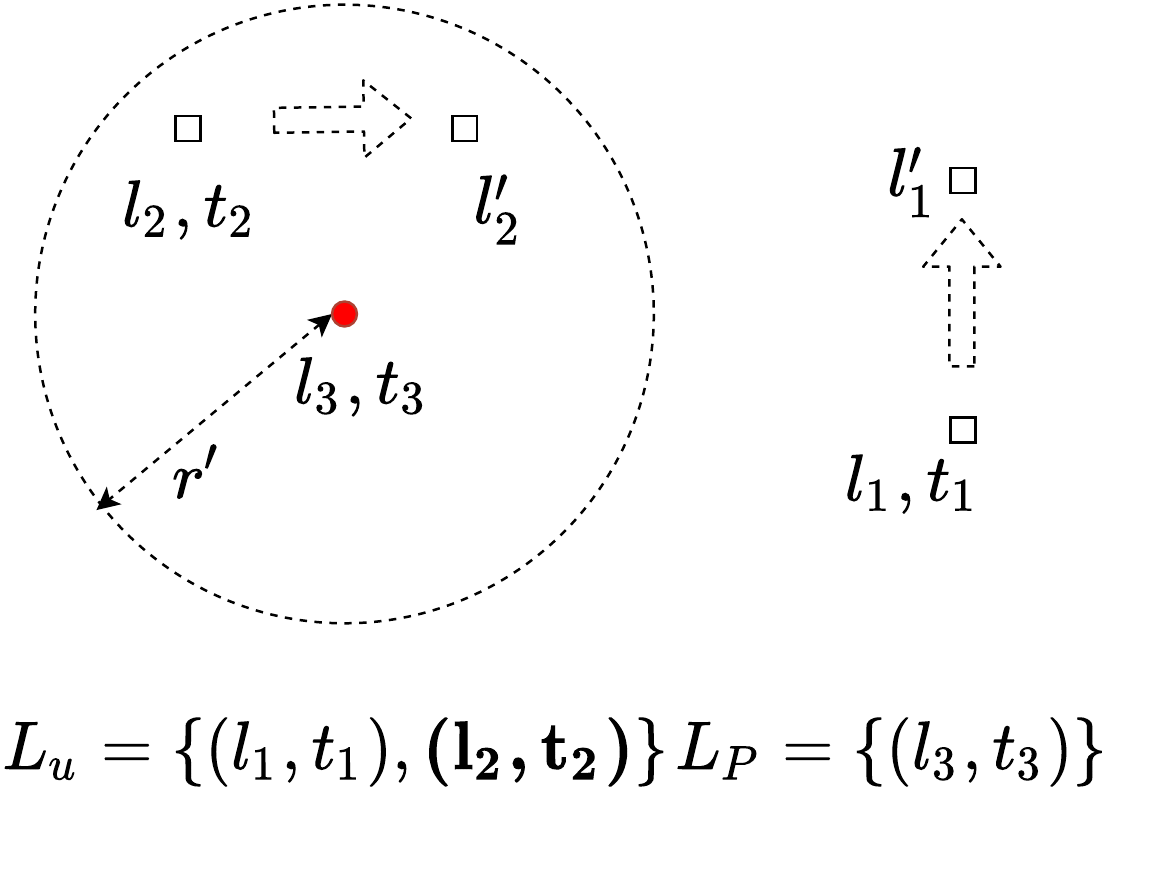}
		   \caption{Subset selection.}
		   \label{subfig:subset}
	   \end{subfigure}
	\caption{\small \small The overall workflow and the subset selection process in Step 3. }	
	\label{fig:workflow_subset}
   \end{figure}

\fakeparagraph{Step 1. Location perturbation}
At this step, the user perturbs his/her visited location set $L_u$ to a perturbed visited location set $L'_u$, and submits $L'_u$ to the server. The steps are similar to the Geo-I baseline (\secref{subsec:geo-i_baseline}). 

Given a privacy budget $\epsilon$ for each user, and the true location set $L_u$, the user generates a perturbed location set $L'_u$ using \algoref{algo:perturb_location_set}. Similar to the Geo-I baseline, the total privacy budget $\epsilon$ is equally divided into $|L_u|$ shares. Each location $l \in L_u$ is then perturbed to a noisy location $l'$ by Geo-I mechanism with a privacy budget of $\epsilon'=\epsilon/|L_u|$. 

Similar to the Geo-I baseline, we abuse notations $L_u$ and $L'_u$ to denote only the locations (without the timestamps) and the perturbed locations. This is slightly different from the original definition for the visited location set $L_u$ in \defref{def:visited_locations}, where each location is associated with a timestamp. We simply ignore all the temporal information during the perturbation process, as a way to provide strong privacy guarantee for the timestamps (as we do not use it at all). The temporal information will be used to check for the temporal constraints for determining a contact in later MPC steps (Step 3). 

\fakeparagraph{Privacy analysis} To understand the level of privacy guarantee of \algoref{algo:perturb_location_set}, we provide the theoretical privacy analysis for it. It extends Geo-I \cite{andres13} from protecting a single location to a set of locations. The brief idea of the extension was mentioned in \cite{andres13}, we provide a concrete and formal analysis here.

First, we extend the privacy definition of Geo-I from a single location to a set of locations. 

\begin{definition} (General-Geo-I) \label{def:geo-i-general}For two tuples of locations $\mathbf{l}=(l_1, \ldots, l_n),\mathbf{l}'=(l'_1, \ldots, l'_n)$, a privacy parameter $\epsilon$, a mechanism $M$ satisfies  $\epsilon$-General-Geo-I iff:
	$$d_{\rho}(M(\mathbf{l}),M(\mathbf{l}'))\leq \epsilon d_{\infty}(\mathbf{l},\mathbf{l}'),$$
\end{definition}
where $d_{\infty}(\mathbf{l},\mathbf{l}') = max_id(l_i, l'_i)$, which is the largest Euclidean distance among all pairs of locations from $\textbf{l}$ and $\textbf{l}'$. 

Different from the Geo-I definition for a single location, % in \defref{def:geo-i}, 
here the distance between two tuples of locations is defined as the largest distance at any dimensions. If a randomized mechanism $M$ satisfies the  General-Geo-I in \defref{def:geo-i-general}, then given the any perturbed location set output $\mathbf{o}=(o_1, \ldots, o_n)$, which is a sequence of perturbed locations generated by $M$, no adversary could distinguish whether the perturbed locations are generated from true location set $\mathbf{l}$ or from $\mathbf{l}'$. The ratio of the probabilities of producing the same output $\mathbf{o}$ from the true location set $\mathbf{l}$ or its neighboring input $\mathbf{l}'$ is bounded by $\epsilon d_{\infty}(\mathbf{l},\mathbf{l}')$. 

As a special case, when there is only one location in $\mathbf{l}$, \defref{def:geo-i-general} becomes the original definition in Geo-I%\defref{def:geo-i}
. We start the privacy analysis by showing the composition theorem for the General-Geo-I. 

\begin{lemma} Let $K_0$ be a mechanism satisfying $\epsilon_1$-Geo-I and $K_1$ be another mechanism satisfying $\epsilon_2$-Geo-I. For a given 2-location tuple $\mathbf{l}=(l_1, l_2)$, the combination of $K_0$ and $K_1$ defined as $K_{1,2}(\mathbf{l})=(K_1(l_1), K_2(l_2))$ satisfies $(\epsilon_1+\epsilon_2)$-General-Geo-I (\defref{def:geo-i-general}). 
\label{lemma:geo-i_compose}
\end{lemma}

\begin{proof}
Let a potential output of the combined mechanism, $K_{1,2}(\mathbf{l})$,  be $\mathbf{o}=(o_1, o_2)$, indicating a tuple of two perturbed locations. The probability of generating $\mathbf{o}$ from two input location sets $\mathbf{l}=(l_1, l_2)$ and $\mathbf{l}'=(l'_1, l'_2)$ are $\Pr[K_{1,2}(\mathbf{l})=\mathbf{o}]$ and $\Pr[K_{1,2}(\mathbf{l}')=\mathbf{o}]$, respectively. Then, we measure the ratio of the two probabilities:

\begin{align}
	 \frac{\Pr[K_{1,2}(\mathbf{l})=\mathbf{o}]}{\Pr[K_{1,2}(\mathbf{l}')=\mathbf{o}]}  & = \frac{\Pr[K_1(l_1)=o_1]\cdot \Pr[K_2(l_2)=o_2]}{\Pr[K_1(l'_1)=o_1]\cdot \Pr[K_2(l'_2)=o_2]} \label{eq:proof_1}\\
	& = \frac{\Pr[K_1(l_1)=o_1]}{\Pr[K_1(l'_1)=o_1]}\cdot \frac{\Pr[K_2(l_2)=o_2]}{\cdot \Pr[K_2(l'_2)=o_2]} \label{eq:proof_2}\\
	& \leq e^{\epsilon_1 \cdot d(l_1, l'_1)} \cdot e^{\epsilon_2 \cdot d(l_2, l'_2)} \label{eq:proof_3}\\
	& \leq e^{\epsilon_1 \cdot d_{\infty}(\mathbf{l}, \mathbf{l}')} \cdot e^{\epsilon_2 \cdot d_{\infty}(\mathbf{l}, \mathbf{l}')} \label{eq:proof_4}\\
	& = e^{(\epsilon_1 + \epsilon_2) d_{\infty}(\mathbf{l}, \mathbf{l}')} \label{eq:proof_5}
\end{align}

The end result in \equref{eq:proof_5} shows that $K_{1,2}$ satisfies $(\epsilon_1+\epsilon_2)$-General-Geo-I (\defref{def:geo-i-general}).

\end{proof}

\begin{theorem}\label{theo:privacy_location_perturb} \algoref{algo:perturb_location_set}  satisfies $\epsilon$-General-Geo-I (\defref{def:geo-i-general}). 
\label{theo:general_geo_i}
\end{theorem}
\begin{proof}

	We use the composition theorem of Geo-I in \lemref{lemma:geo-i_compose} to show that \algoref{algo:perturb_location_set} composes linearly and consumes a total privacy budget of $\epsilon$ over the entire set of locations. 

	At \lineref{line:geo_i} of \algoref{algo:perturb_location_set}, for each location $l\in L_u$,  it is perturbed to $l'$ with a privacy budget $\epsilon' = \epsilon_u/|L_u|$. Then, for each location, we achieve $\epsilon'$-Geo-I. If there are two locations in $L_u$, and because each location is perturbed independently, we achieve $\epsilon' + \epsilon' = 2\epsilon'$-General-Geo-I. Then, by induction, for $|L_u|$ locations, we achieve $\sum_{l\in L_u} \epsilon' = \epsilon' \cdot |L_u| = \frac{\epsilon}{|L_u|}\cdot |L_u| =\epsilon$-General-Geo-I. 

\end{proof}

\fakeparagraph{Step 2. Subset selection}
After Step 1, each user has a set of perturbed locations $L'_u$ and submits it to the server. At Step 2, subset selection, upon receiving $L'_u$, the server selects a subset of high-risk locations based on the perturbed locations. This step is done at the server side, where the patients' true locations are stored. The server returns to the client the index of the selected high-risk locations (\eg if the 2nd location and the 4th location are high-risk locations, then the server returns $\{2, 4\}$) to indicate which location out of $L_u$ are close to the patients. In order to protect the privacy of the patients, randomized response is used to perturb the indexes of high-risk locations to noisy indexes. 

On the server side, the server does not know the true locations nor the timestamps of the visits, which are stored in $L_u$ on the client side (the user). However, the server has access to the true locations visited by the patients, which are denoted as $L_P$. $L_P$ are the basis to judge whether a user location is high-risk or not. 

%A detailed running example is deferred to the full report \cite{full-report}. 

\begin{theorem}\label{theo:privacy_subset} The subset selection step satisfies $\epsilon_P$-Local Differential Privacy . 
\end{theorem}

We also provide a detailed running example therein for this step.

\fakeparagraph{Step 3. Accelerated MPC}
After Step 2, the user $u$ receives the set $\tilde{I}$, which contains the indexes of high-risk locations as identified by the server. At Step 3, accelerated MPC, the user uses MPC protocol to check with the server about whether the user is a contact or not, using only the locations indicated by the set $\tilde{I}$. In contrast, the MPC baseline uses every location in $L_u$ to check with the server using MPC protocol. Here, the user uses only a small subset of locations to perform secure computations. 

Our method runs the same secure computation procedure as the MPC baseline. However, it only provides the subset of locations as the inputs to the \texttt{ProtocolIO} object \texttt{io}. The timestamps are also included in the $io$ object, and are used to compare with the patients' trajectories $L_P$. This is different from the Geo-I baseline or Step 1. location perturbation in ContactGuard, which only uses the spatial information (the locations).

\fakeparagraph{Overall Analysis}
Overall, for the server to finish processing all users, it takes $\mathcal{O}(\sum_u |I_u||L_P|) \to \mathcal{O}(|U||L_P|\max_u|I_u|)$, where $I_u$ indicates the subset for each user $u$. 

\fakeparagraph{Privacy analysis} ContactGuard offers end-to-end privacy guarantee \wrt to both $L_u$ and $L_P$.
\fakeparagraph{}

\tabref{tab:comparison_methods} lists the comparison of the two baselines with our proposed ContactGuard. While achieving the same privacy and accuracy, our proposed solution significantly improves the efficiency of the MPC solution. 

\begin{table}
	\centering 
	{\small\scriptsize
		\caption{\small Comparison of baselines and ContactGuard} \label{tab:comparison_methods}
		\begin{tabular}{l|p{3cm}|p{3.2cm}|p{3cm}}
			\hline
			{\bf Methods} & {\bf MPC baseline} & {\bf Geo-I baseline} & {\bf ContactGuard} \\ 
			\hline 
			Accuracy (Recall) & 100\% $\surd$ & 39.0-70\% $\times$ & 86-100\% $\surd$\\
			\hline
			Efficiency & \# of Secure Operations: $\mathcal{O}(|U||L_P|\max_u|L_u|)$ $\times$ & \#of Plaintext Operations: $\mathcal{O}(|U||L_P|\max_u|L_u|)$ $\surd$ & \# of Secure Operations: $\mathcal{O}(|U||L_P|\max_u|I_u|)$ $\surd$\\
			\hline
			Privacy & Confidentiality $\surd$ & Local-$d$ DP $\surd$ & Local-$d$ DP $\surd$\\
			\hline

		\end{tabular}
	}
\end{table}

\section{Experimental Study}
\label{sec:experiment}

In this section, we first introduce the experimental setup in Sec.~\ref{subsec:experimentSetup}. Then, we present the detailed experimental results in Sec.~\ref{subsec:experimentResults}.

\subsection{Experimental setup}
\label{subsec:experimentSetup}

\fakeparagraph{Datasets} We use a real-world dataset Gowalla \cite{DBLP:conf/kdd/ChoML11} and a randomly generated synthetic dataset in our experiments. 

\fakeparagraph{Baselines} We compare our proposed ContactGuard (short as \textbf{CG}) method with the MPC baseline (\secref{subsec:mpc_baseline}, short as \textbf{MPC}) and the Geo-I baseline (\secref{subsec:geo-i_baseline}, short as \textbf{Geo-I}). 

\fakeparagraph{Metrics}
\label{subsec:experimentsMetrics}
We focus on the following metrics:
\begin{itemize}
\setlength\itemsep{0.1em}
    \item Recall: the number of found true close-contacts divided by the total number of close-contacts. 
    \item Precision: the number of found true close-contacts divided by the number of predicted close-contacts. 
    \item $F_1$ score: the harmonic mean of the precision and recall. $F_1=2\cdot \frac{\text{precision}\cdot \text{recall}}{\text{precision}+ \text{recall}}$.
    \item Accuracy: the percentage of correctly classified users. 
    \item Running time: the running time of the method in seconds. 
\end{itemize}

Recall is considered as the most important metric in our experiments, as it measures the capability of the tested method for identifying potential close-contacts out of the true close-contacts, which is crucial in the application scenario of contact tracing of infectious disease. 

\fakeparagraph{Control variables}
We vary the following control variables in our experiments to test our method. Default parameters are in boldface. The privacy budget setting adopts the common setting with the real-world deployments (Apple and Google) \cite{domingo2021limits}.
\begin{itemize}
\setlength\itemsep{0.1em}
    \item  Number of users: $|U| \in [ \textbf{200}, 400, 800, 1600]$. 
    \\Scalability test: $|U| \in [10K, 100K, 1M]$. 
    \item Privacy budget for each user: $\epsilon \in [2.0, 3.0, \textbf{4.0}, 5.0]$.  
    \item Privacy budget for the patients: $\epsilon_P \in [2.0, 3.0, \textbf{4.0}, 5.0]$.

\end{itemize}

\subsection{Experimental Results}
\label{subsec:experimentResults}

Our results show that the proposed solution ContactGuard achieves the best effectiveness/efficiency tradeoff, reducing the running time of the MPC baseline by up to $2.85 \times$. The efficient Geo-I baseline, however, only provides poor effectiveness.  Meanwhile, ContactGuard maintains the same level of effectiveness (measured by recall, precision, $F_1$ score and accuracy) as the one of the MPC baseline. The scalability tests also show that ContactGuard accelerates the MPC baseline significantly. When the number of users is large (\eg 500K), our ContactGuard supports daily execution, which is desirable in real-world application, while the MPC baseline fails to terminate within reasonable time (\eg 24 hours).

\fakeparagraph{Effectiveness vs. Efficiency tradeoff}
\figref{fig:exp_recall_prec_time} demonstrates the effectiveness and efficiency tradeoff of different methods. The y-axis are the running time of different methods (MPC baseline, Geo-I baseline and our proposed ContactGuard), and the x-axis are the respective measures for effectiveness (recall, precision, $F_1$ and accuracy). 

As we could see from \figref{fig:exp_recall_prec_time}, the MPC baseline is an \textit{exact} method, which obtains 100\% in all measures of effectiveness. However, the high accuracy comes with prohibitive running time. It runs for 99.35 seconds when the number of users is 200 ($|U|=200$) and 194.99 seconds for 400 users ($|U|=400$). 

As a comparison, the Geo-I method is efficient, but comes with a significant decrease in effectiveness. For the setting $|U|=200, \epsilon=4.0$, Geo-I obtains about 70\% recall and merely 46.7\% precision. For the setting $|U|=400, \epsilon=3.0$, the recall drops to 39.0\% and the precision drops to 42.0\%. 

Our proposed solution ContactGuard achieves the best effectiveness/efficiency tradeoff. For the setting $|U|=200, \epsilon=4.0$ (\figref{fig:exp_recall_prec_time}), CG obtains the same level of effectiveness as the one of the MPC baseline, achieving 100\% in recall, precision and accuracy. 

As compared to the Geo-I baseline, CG method obtains 2.28$\times$, 2.38$\times$, and 1.24$\times$ improvement in terms of recall, precision and accuracy, respectively. On the other hand, when we compare to the \textit{exact} MPC baseline, the CG method is about 2.5$\times$ faster.  

\begin{figure}[t]
    \centering
       \begin{subfigure}[b]{0.23\textwidth}
           \centering
           \includegraphics[width=\textwidth]{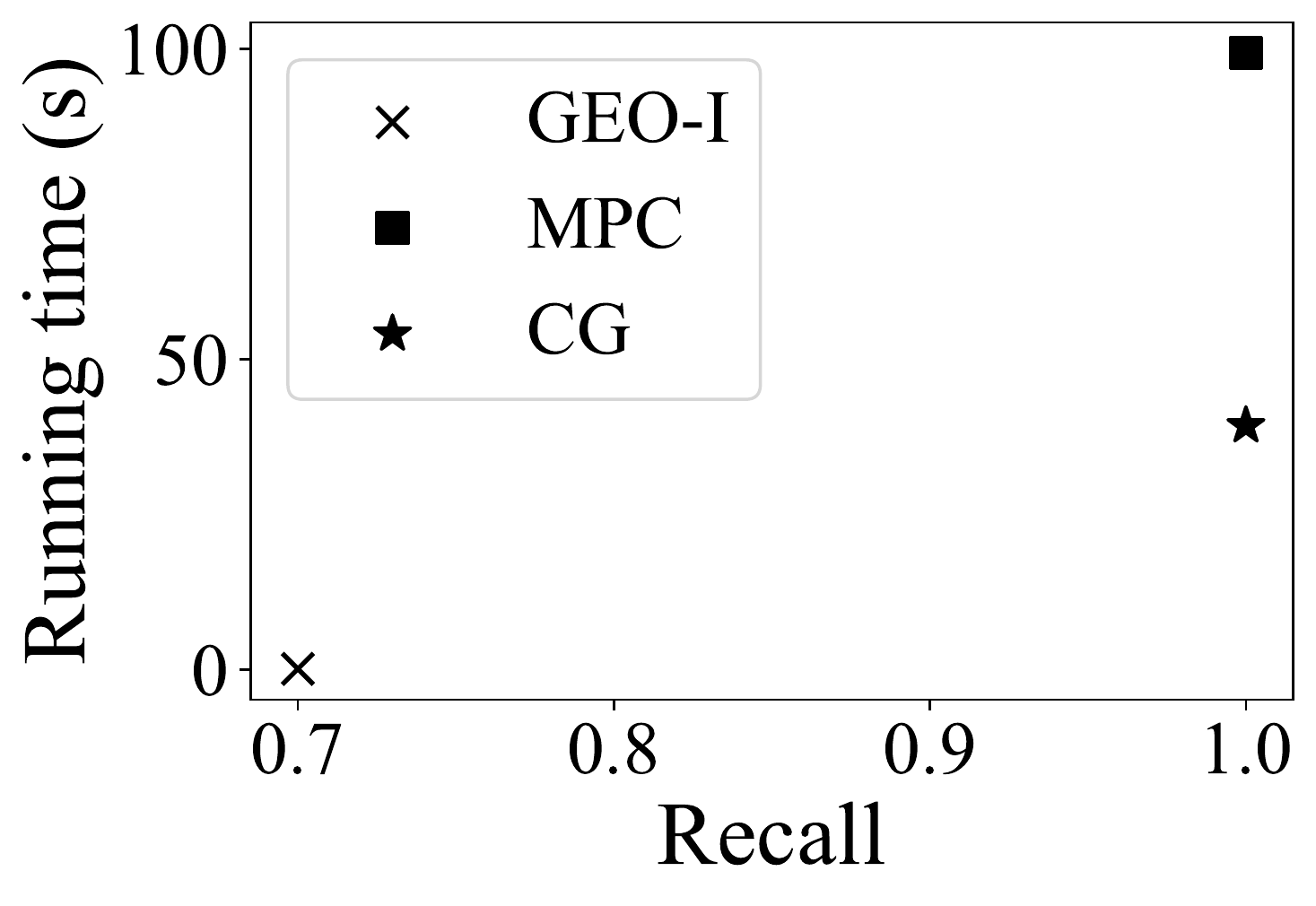}
           \caption{Recall}
           \label{subfig:exp_recall_time}
       \end{subfigure}
       \hfill
        \begin{subfigure}[b]{0.23\textwidth}
           \centering
           \includegraphics[width=\textwidth]{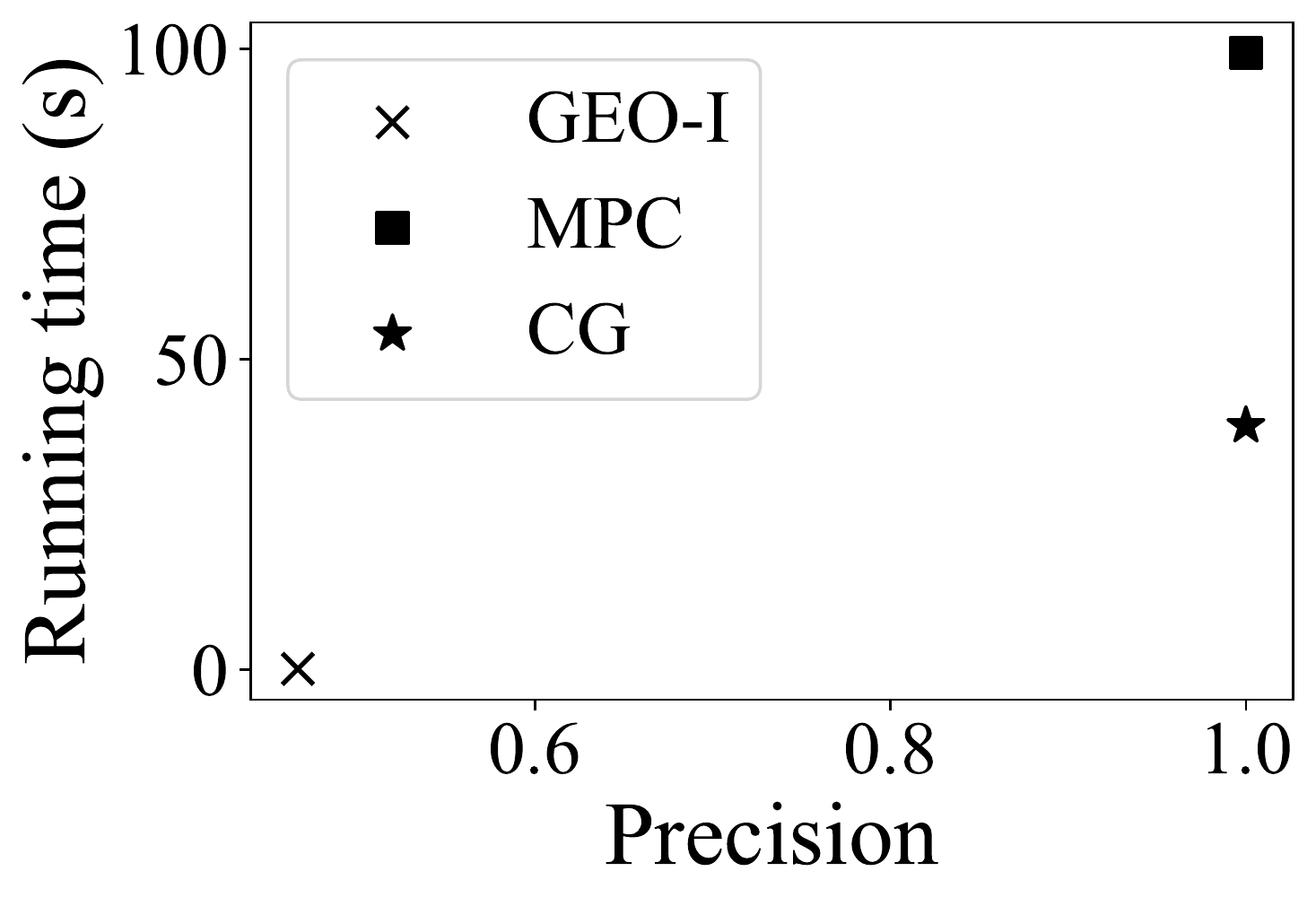}
           \caption{Precision}
           \label{subfig:exp_prec_time}
       \end{subfigure}
       \hfill
        \begin{subfigure}[b]{0.23\textwidth}
           \centering
           \includegraphics[width=\textwidth]{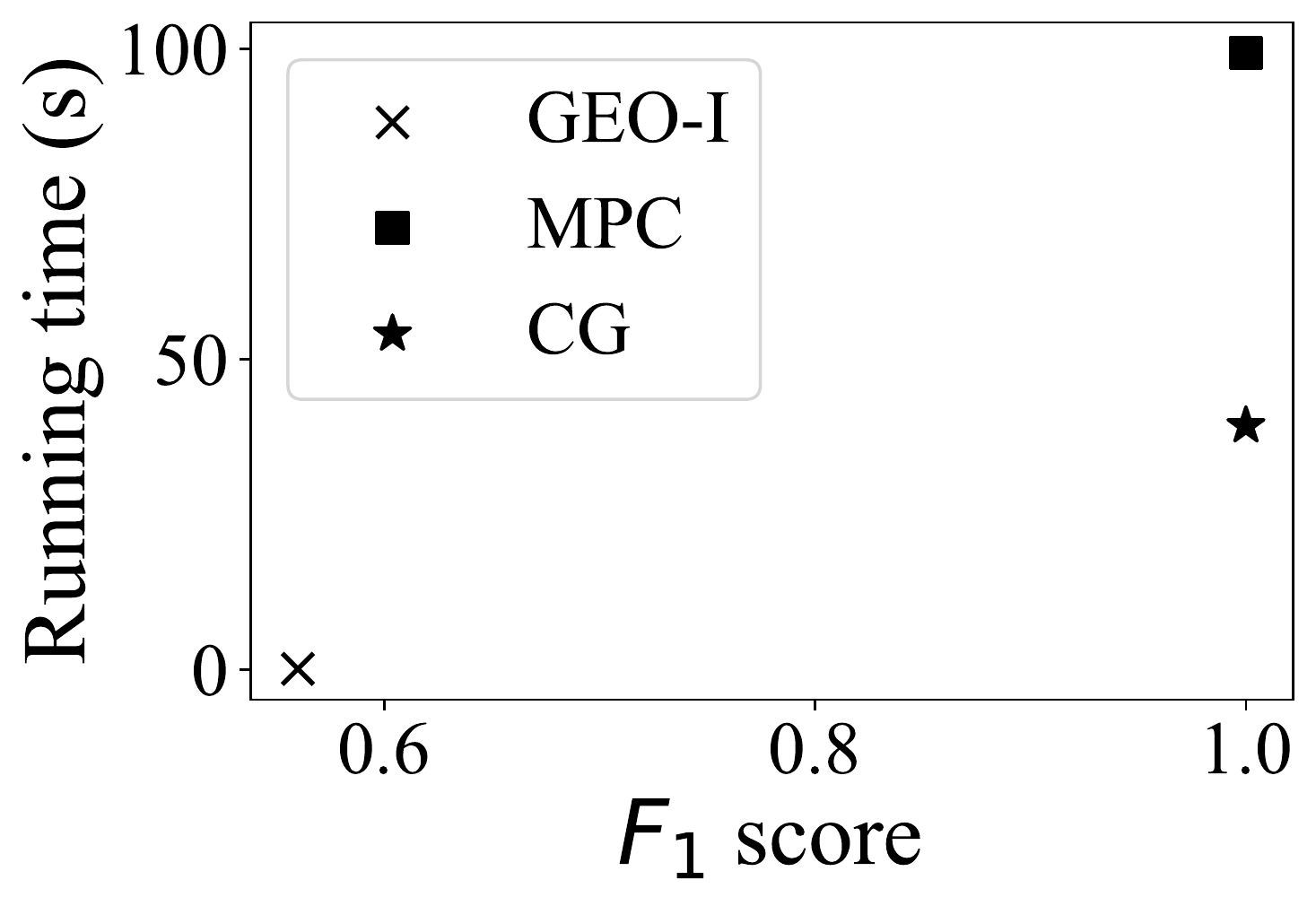}
           \caption{$F_1$ score}
           \label{subfig:exp_f1_time}
       \end{subfigure}
       \hfill
        \begin{subfigure}[b]{0.23\textwidth}
           \centering
           \includegraphics[width=\textwidth]{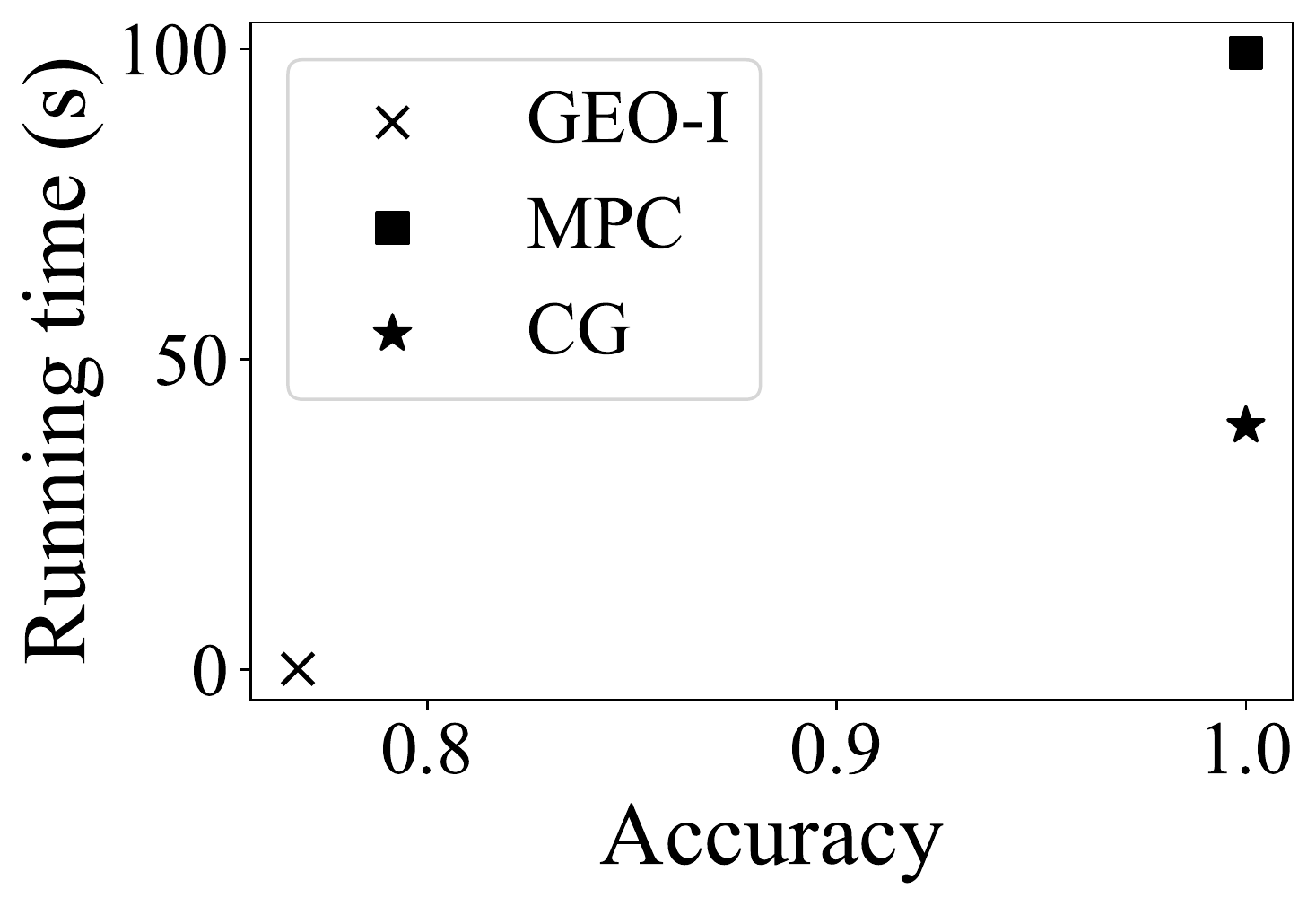}
           \caption{Accuracy}
           \label{subfig:exp_acc_time}
       \end{subfigure}
    \caption{\small Effectiveness (Recall/Precision/$F_1$/Accuracy) vs. Running time trade-off of different methods on the Gowalla dataset.  $ |U|=200, \epsilon=4.0.$}\label{fig:exp_recall_prec_time}
   \end{figure}

\fakeparagraph{Effect of control variables}
Then, we test the effectiveness across different input sizes ($|U|$) and different values of the privacy budget ($\epsilon$). 

Varying $|U|$: \figref{fig:exp_recall_U} shows measures of effectiveness (recall, precision, $F_1$, and accuracy) when we vary the number of users ($|U|$) for the Gowalla dataset. Across different input sizes, the CG method maintains the same level of recall/precision/$F_1$ score/accuracy as the \textit{exact} MPC baseline across different $|U|$. This verifies that our CG method is highly effective across different input sizes. 

On the other hand, the Geo-I baseline sacrifices significant effectiveness (in recall/precisio\\n/$F_1$/accuracy) across different $|U|$. The recall drops from 70\% when $|U|=200$ to 54.2\% when $|U|=1600$. The precision is constantly below 50\%, and hovers around 35\% when $|U|$ gets larger than 800. It shows that the pure differential privacy based solution injects noise to the perturbed locations, and the perturbed locations inevitably lead to poor effectiveness in the contact tracing applications. 

Varying $\epsilon$: \figref{fig:exp_effectiveness_eps} demonstrates the impact of the privacy budget on the effectiveness (recall/precision/$F_1$/accuracy). In terms of recall, CG obtains the same level of effectiveness (reaching 100\%) as the MPC baseline when $\epsilon=5.0$. The recall is 85.2\% when the privacy budget is small ($\epsilon=2.0$) and improves to 88.9\% when $\epsilon=3.0$. In these two settings, the improvement over the Geo-I baseline are about 2.12$\times$. 

In terms of precision, CG never introduces false positives, because positive results are returned only when there is indeed a true visited location of the user (included in the subset selected and verified with MPC operation) which overlaps with the patients. Thus, the precision is constantly at 100\% over different $\epsilon$. 

The $F_1$ score reflects the similar pattern as the recall, as it is an average of recall and precision. The accuracy of CG stays at more than 99\% across different $\epsilon$. \looseness=-1

\begin{figure}[t]
    \centering
       \begin{subfigure}[b]{0.23\textwidth}
           \centering
           \includegraphics[width=\textwidth]{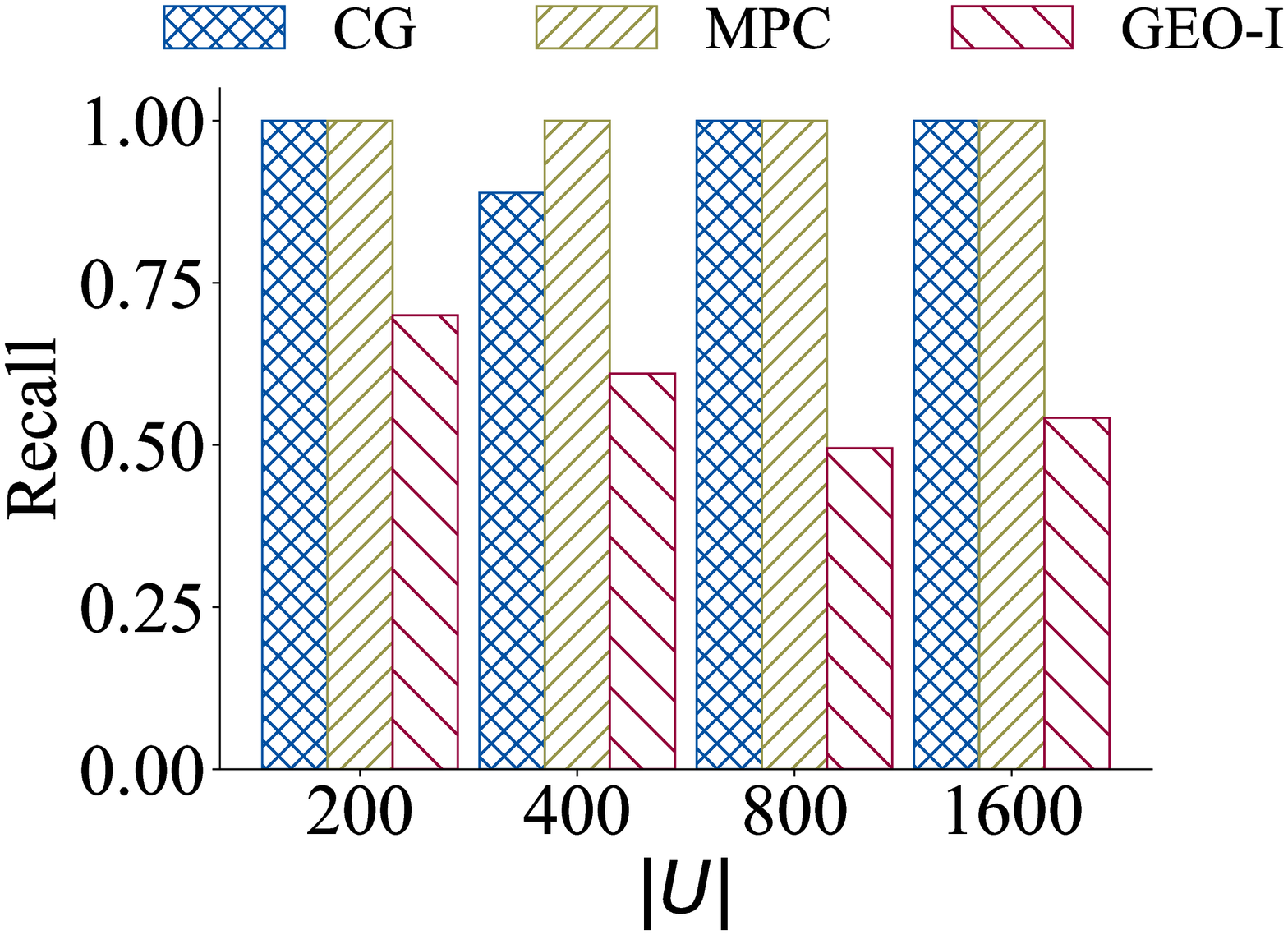}
           \caption{Recall.}
           \label{subfig:exp_recall_U}
       \end{subfigure}
       \hfill
        \begin{subfigure}[b]{0.23\textwidth}
           \centering
           \includegraphics[width=\textwidth]{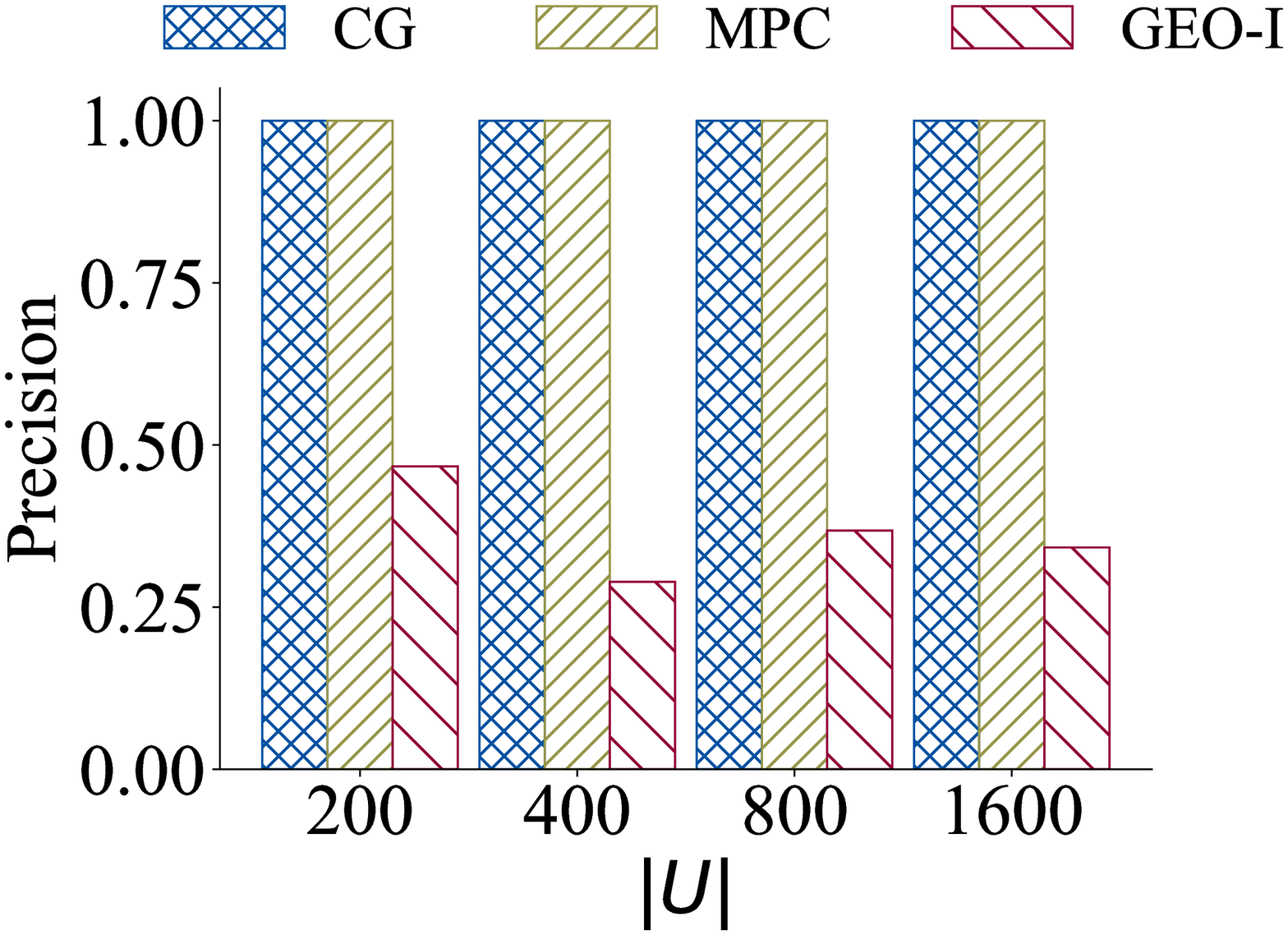}
           \caption{Precision.}
           \label{subfig:exp_prec_U}
       \end{subfigure}
       \hfill
        \begin{subfigure}[b]{0.23\textwidth}
           \centering
           \includegraphics[width=\textwidth]{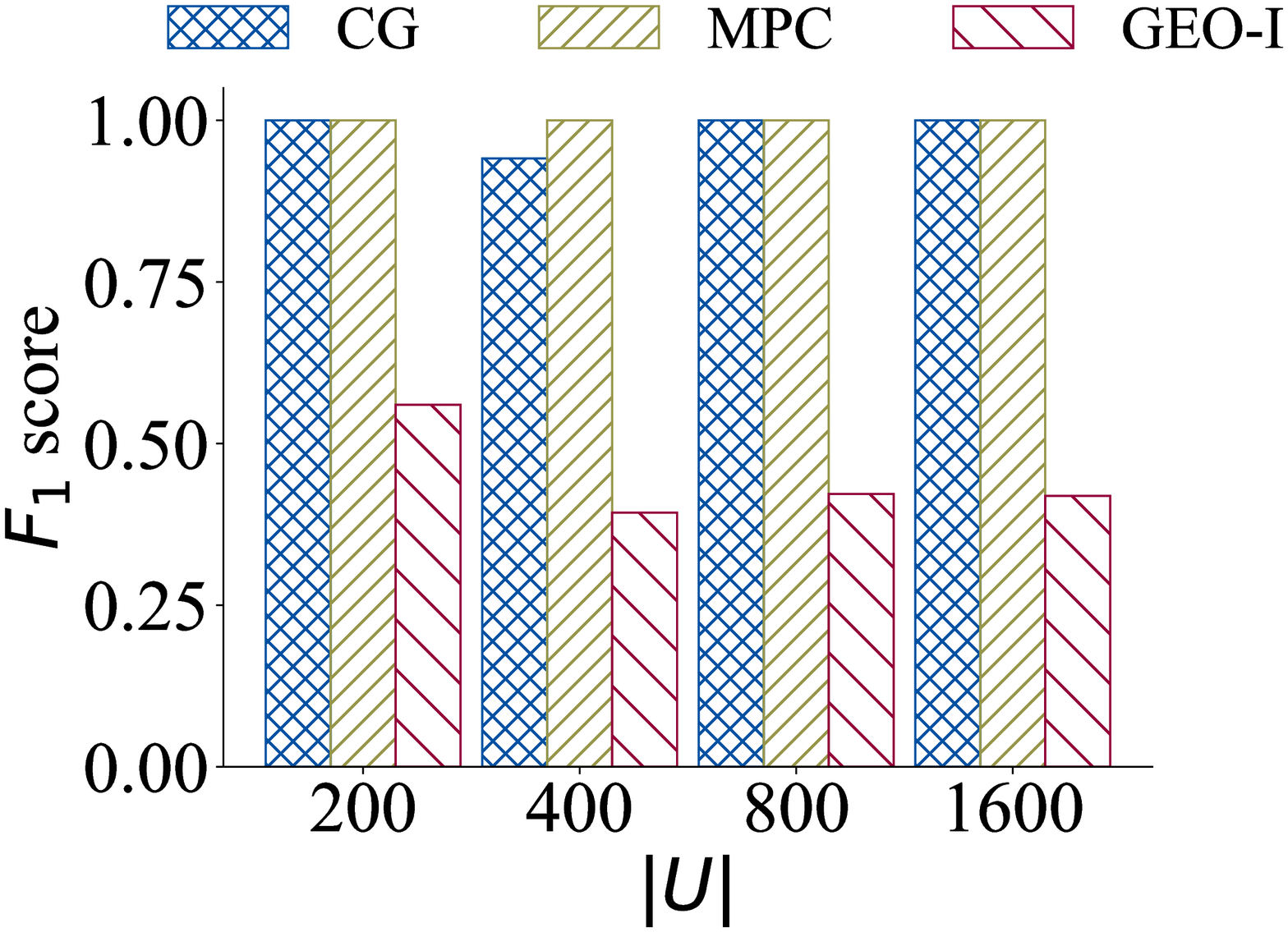}
           \caption{$F_1$ score.}
           \label{subfig:exp_f1_U}
       \end{subfigure}
       \hfill
        \begin{subfigure}[b]{0.23\textwidth}
           \centering
           \includegraphics[width=\textwidth]{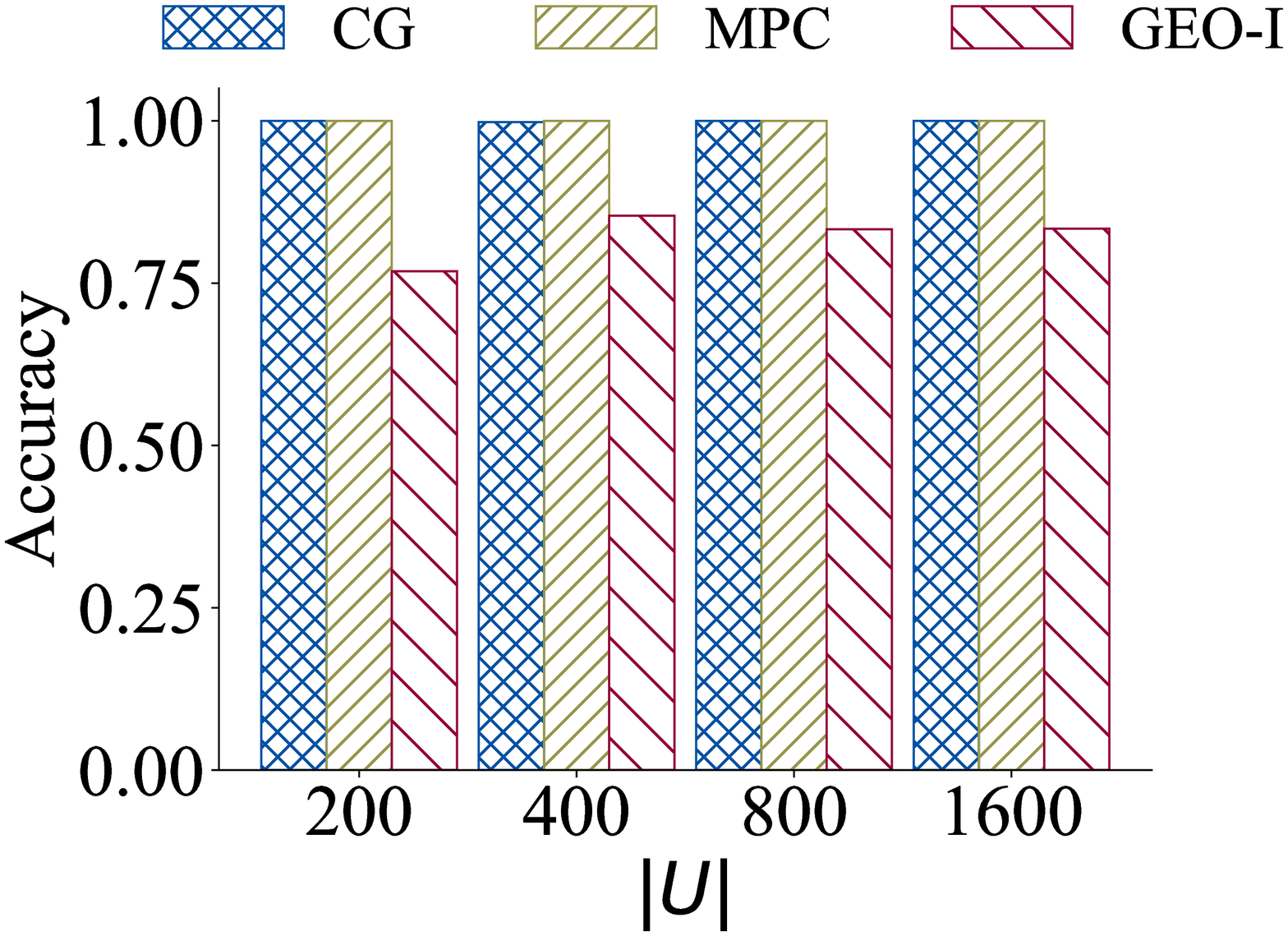}
           \caption{Accuracy.}
           \label{subfig:exp_acc_U}
       \end{subfigure}
    \caption{\small Effectiveness (Recall/Precision/$F_1$/Accuracy) when varying the number of users ($|U|$) on the Gowalla dataset.  $ \epsilon=4.0.$}\label{fig:exp_recall_U}
\end{figure}

\begin{figure}[t]
    \centering
       \begin{subfigure}[b]{0.23\textwidth}
           \centering
           \includegraphics[width=\textwidth]{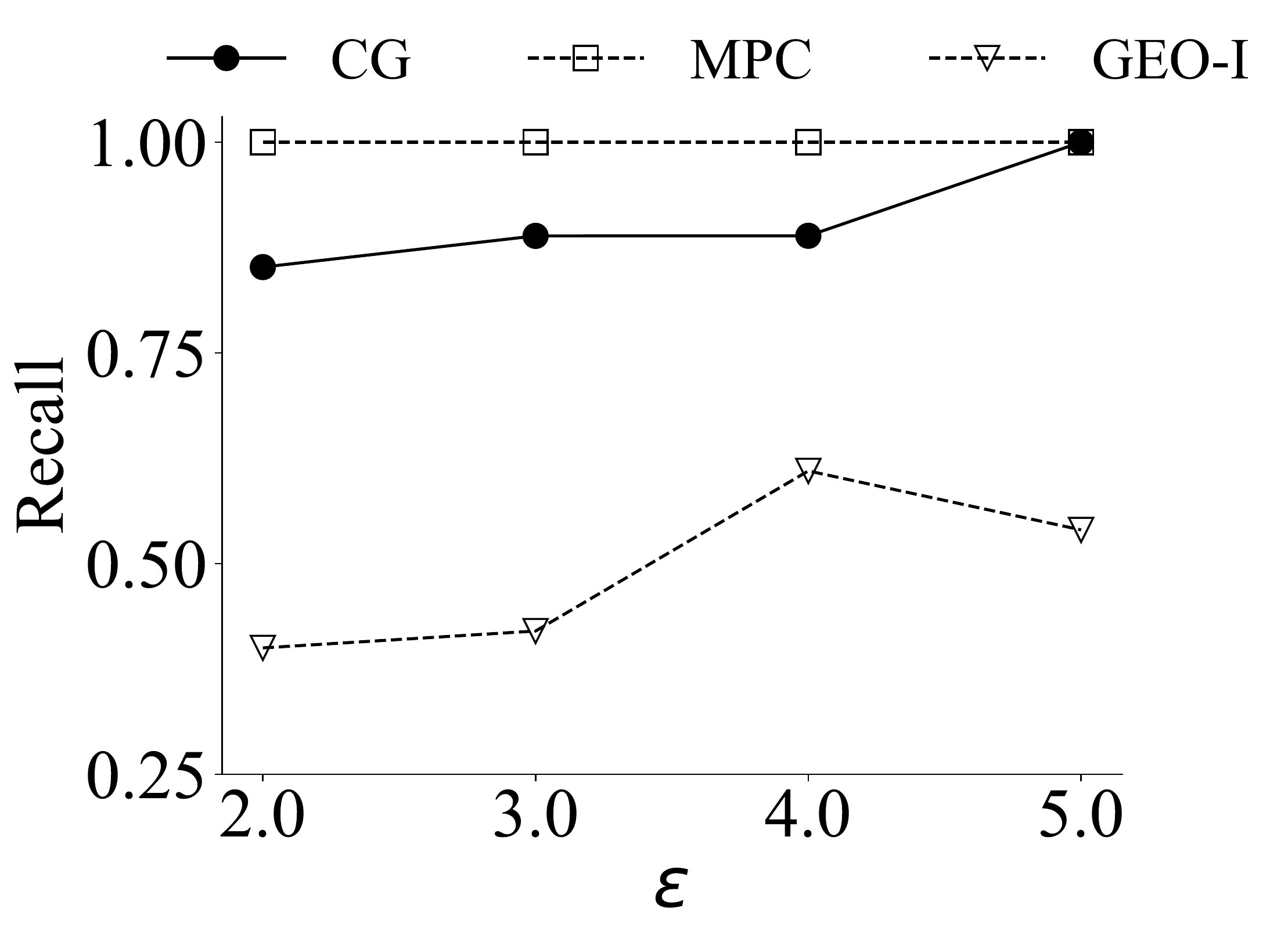}
           \caption{Recall.}
           \label{subfig:recall_eps}
       \end{subfigure}
       \hfill
        \begin{subfigure}[b]{0.23\textwidth}
           \centering
           \includegraphics[width=\textwidth]{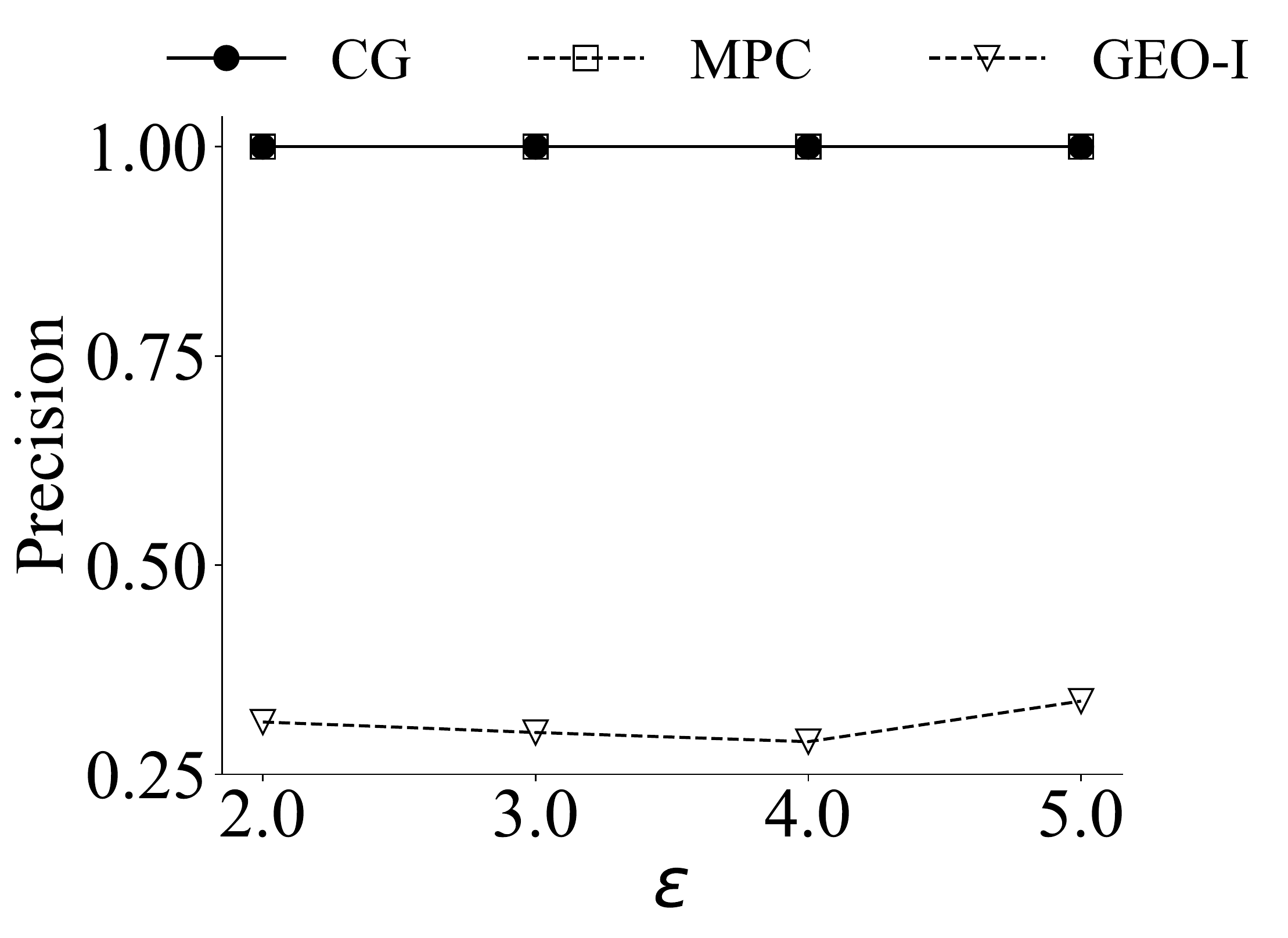}
           \caption{Precision.}
           \label{subfig:prec_eps}
       \end{subfigure}
       \hfill
        \begin{subfigure}[b]{0.23\textwidth}
           \centering
           \includegraphics[width=\textwidth]{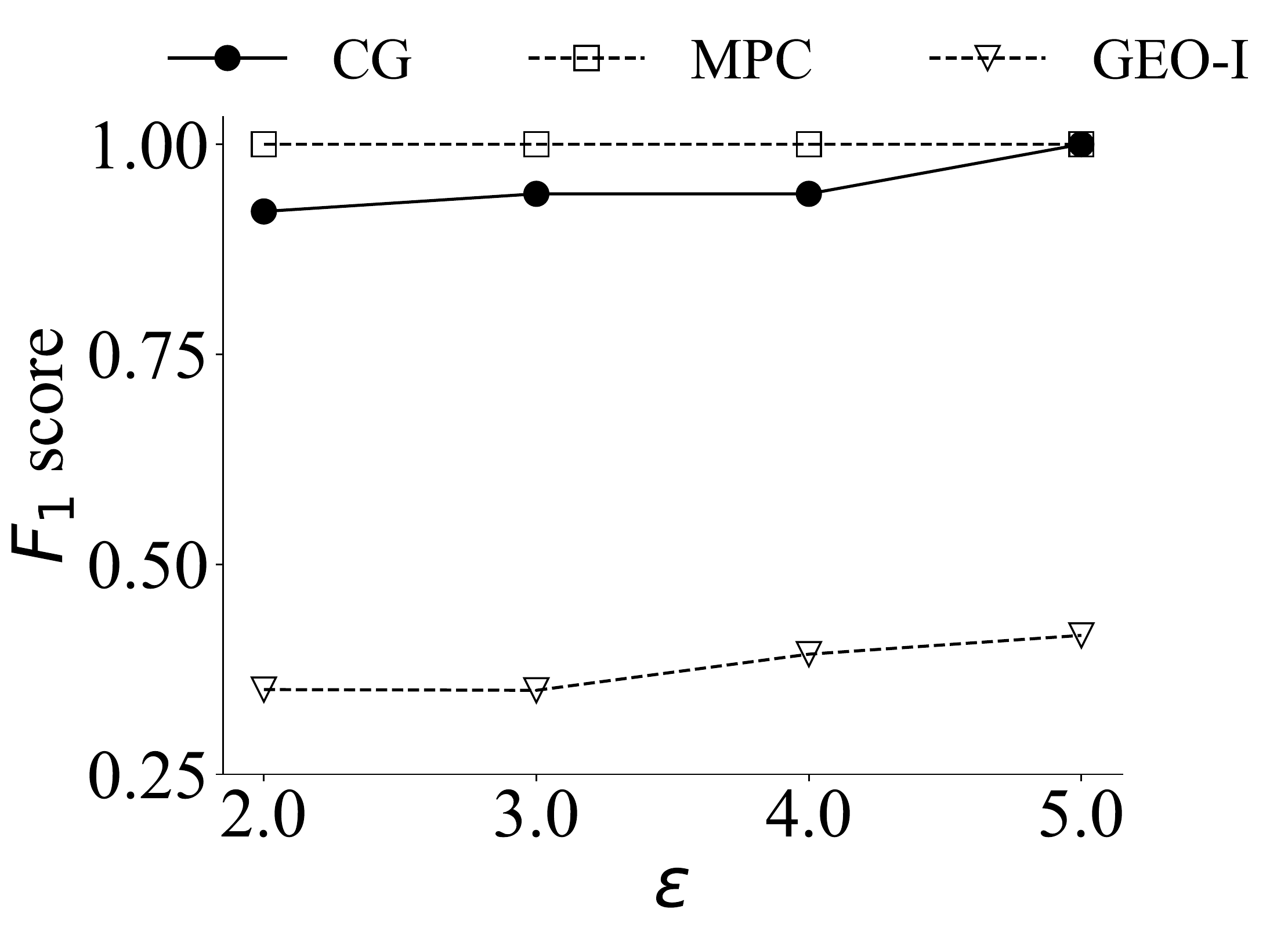}
           \caption{$F_1$ score.}
           \label{subfig:f1_eps}
       \end{subfigure}
       \hfill
        \begin{subfigure}[b]{0.23\textwidth}
           \centering
           \includegraphics[width=\textwidth]{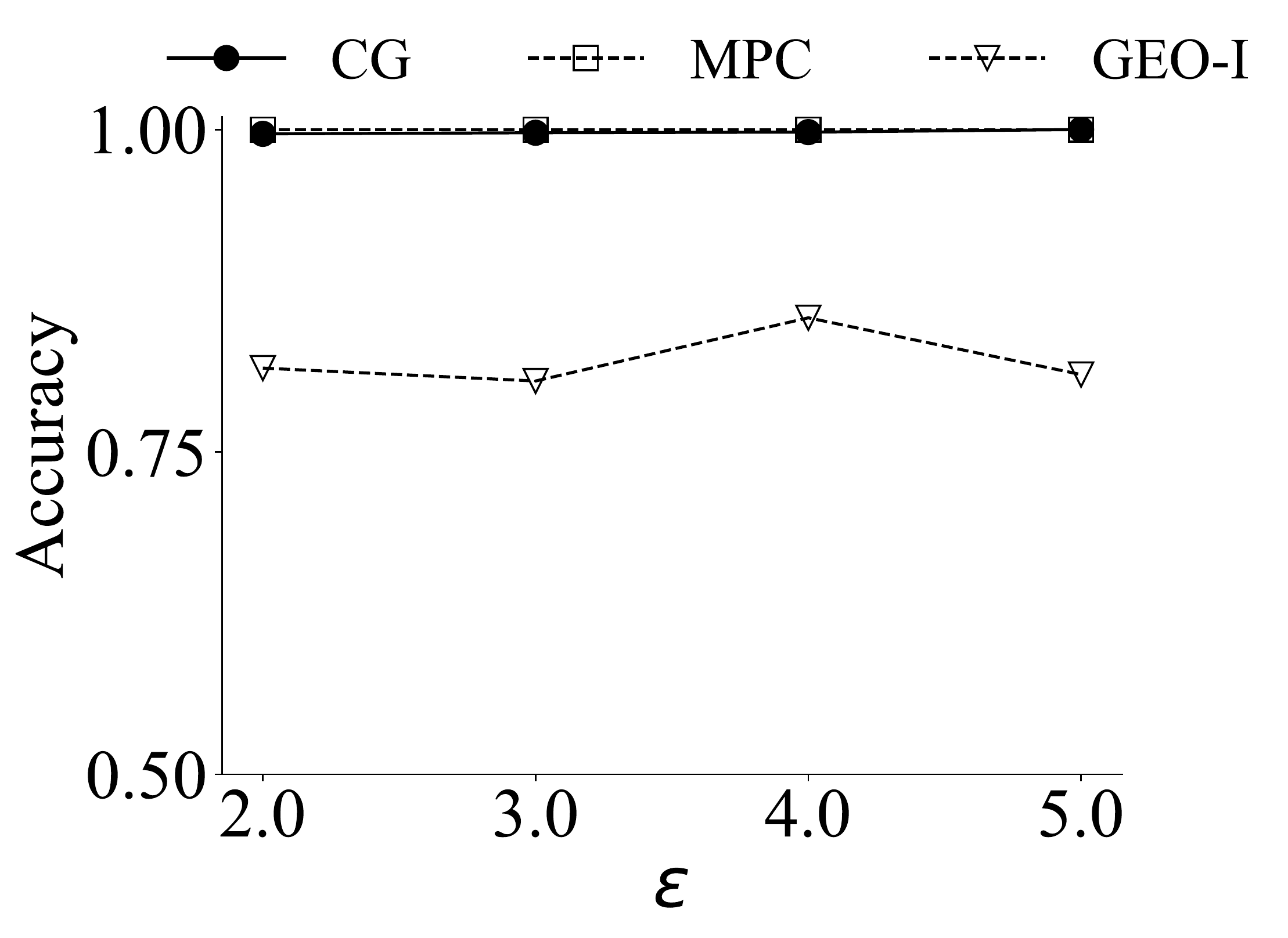}
           \caption{Accuracy.}
           \label{subfig:acc_eps}
       \end{subfigure}
    \caption{\small Effectiveness (Recall/Precision/$F_1$/Accuracy) when varying the privacy budget ($\epsilon$) on the Gowalla dataset.  $ |U|=400.$}\label{fig:exp_effectiveness_eps}
\end{figure}

\fakeparagraph{Efficiency and scalability}
When compared to the MPC baseline, our solution CG achieves significant speed up. \figref{fig:exp_time_U} shows the running time of different methods across different input sizes, when the number of users scales up to 1 million. 

When the number of users ($|U|$) increases from 200 to 1600, the running time of the MPC baseline grows from 99.35 seconds to 746.34 seconds. The running time grows linearly with $|U|$. 

In comparison, the running time of CG grows from 39.31 seconds to 306.20 seconds, also with a linear growth with $|U|$. The speed up of CG over MPC is 2.52$\times$, 2.50$\times$, 2.48$\times$, and 2.44$\times$ when $|U|$ increases from 200 to 1600. The results are shown in \figref{subfig:exp_time_U}. 

Then, we use the synthetic dataset to compare the scalability of the MPC baseline and the CG method. The results are shown in \figref{subfig:exp_time_scale}.  The running time for MPC is 4754.7 seconds, 46,913 seconds (13 hours) and 476,504.6 seconds (5.5 days) for $|U|=$10K, 100K, and 1M.  As a comparison, the CG method runs in 1936.8 seconds, 19,087.2 seconds (5.3 hours) and 195,060.7 seconds (2.25 days). It maintains a $2.5\times$ speed up over the MPC baseline. This verifies that the subset selection in CG significantly reduces the running time of the MPC protocol. When the number of users reaches 1 million, the MPC baseline takes about 5.5 days to process all the users, while our proposed CG finishes in about 2 days.

\begin{figure}[t]
    \centering
       \begin{subfigure}[b]{0.23\textwidth}
           \centering
           \includegraphics[width=\textwidth]{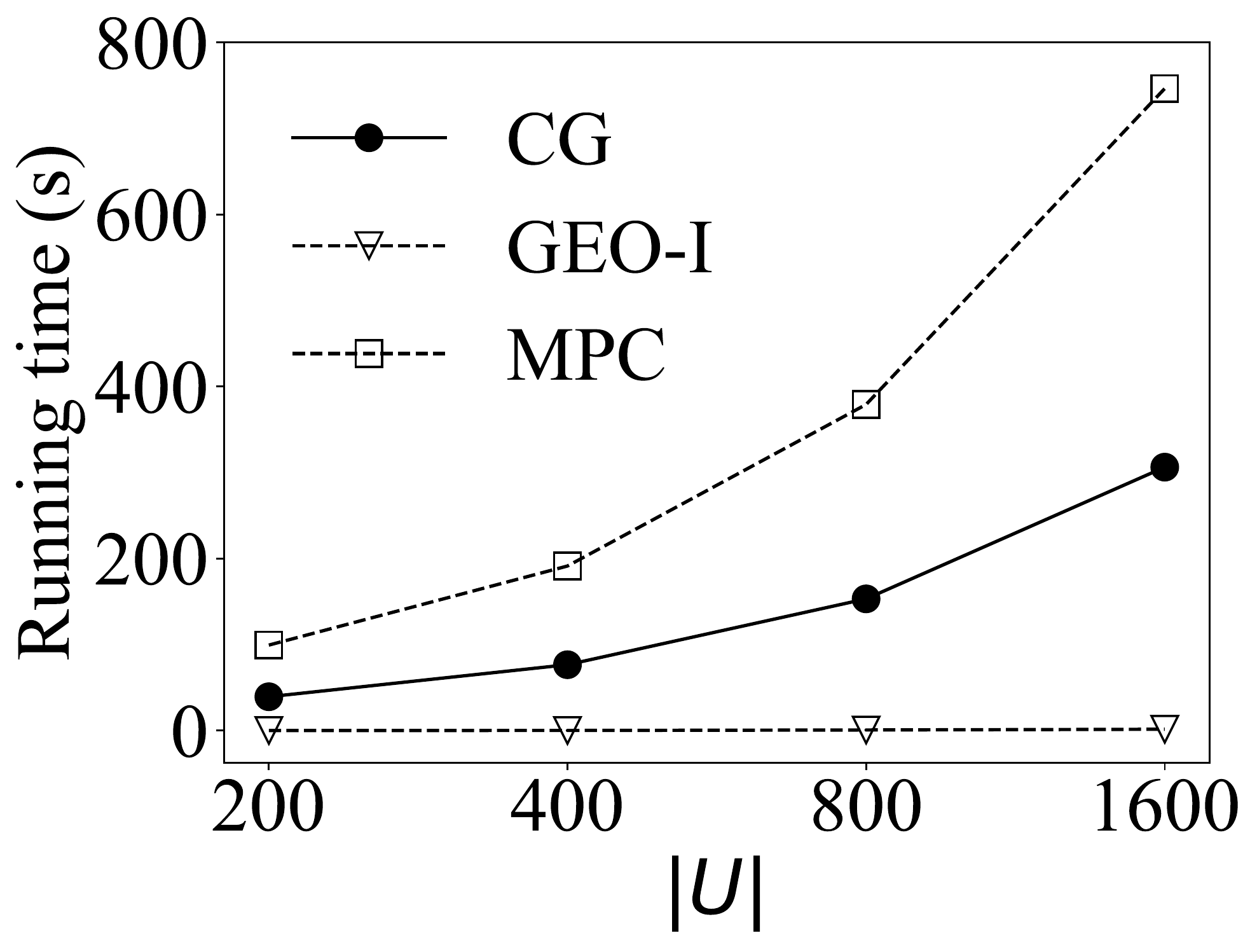}
           \caption{ $\epsilon=4.0$.}
           \label{subfig:exp_time_U}
       \end{subfigure}
       \hfill
        \begin{subfigure}[b]{0.23\textwidth}
           \centering
           \includegraphics[width=\textwidth]{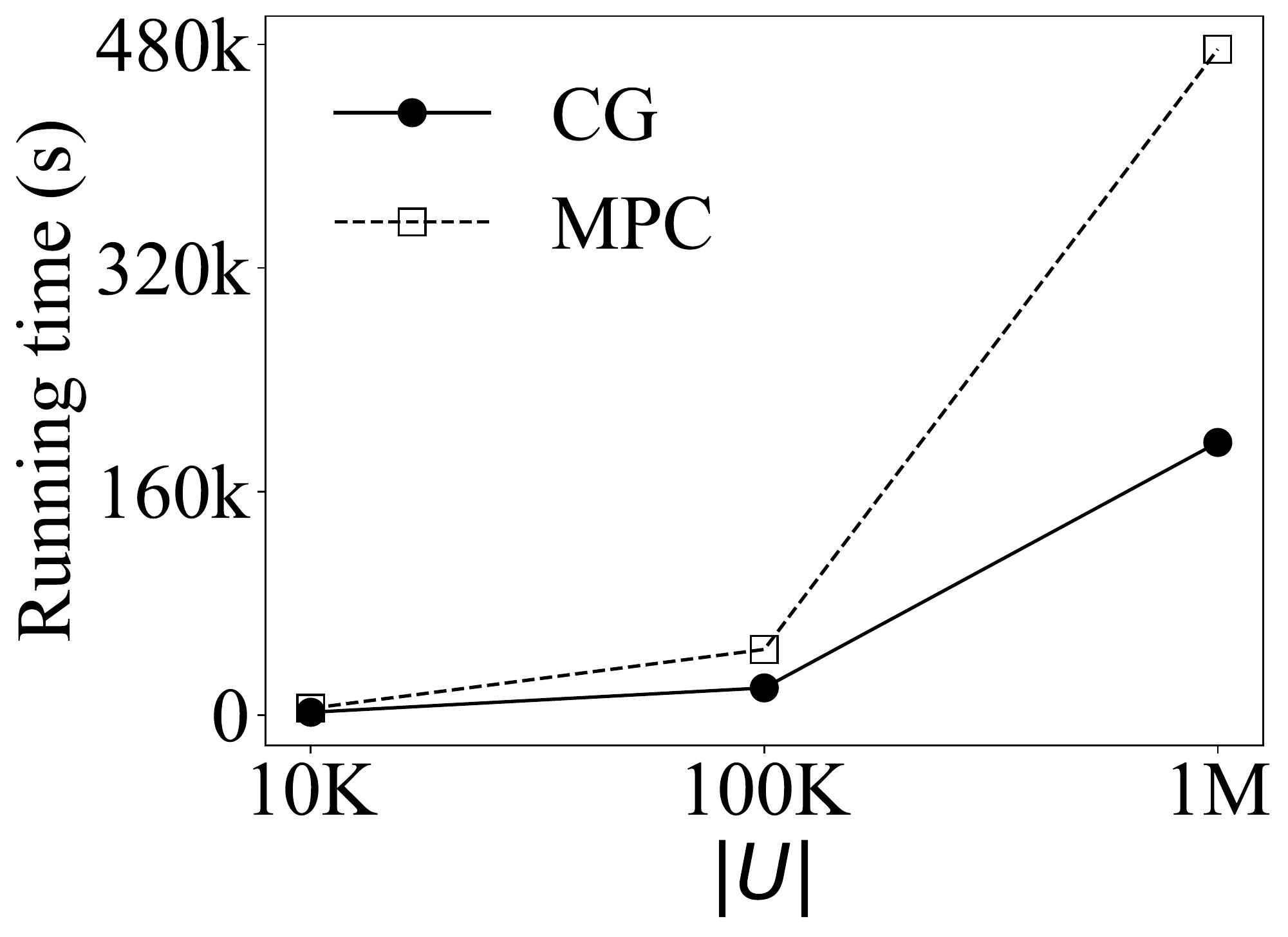}
           \caption{Scalability tests.}
           \label{subfig:exp_time_scale}
       \end{subfigure}
       \hfill
       \begin{subfigure}[b]{0.23\textwidth}
        \centering
        \includegraphics[width=\textwidth]{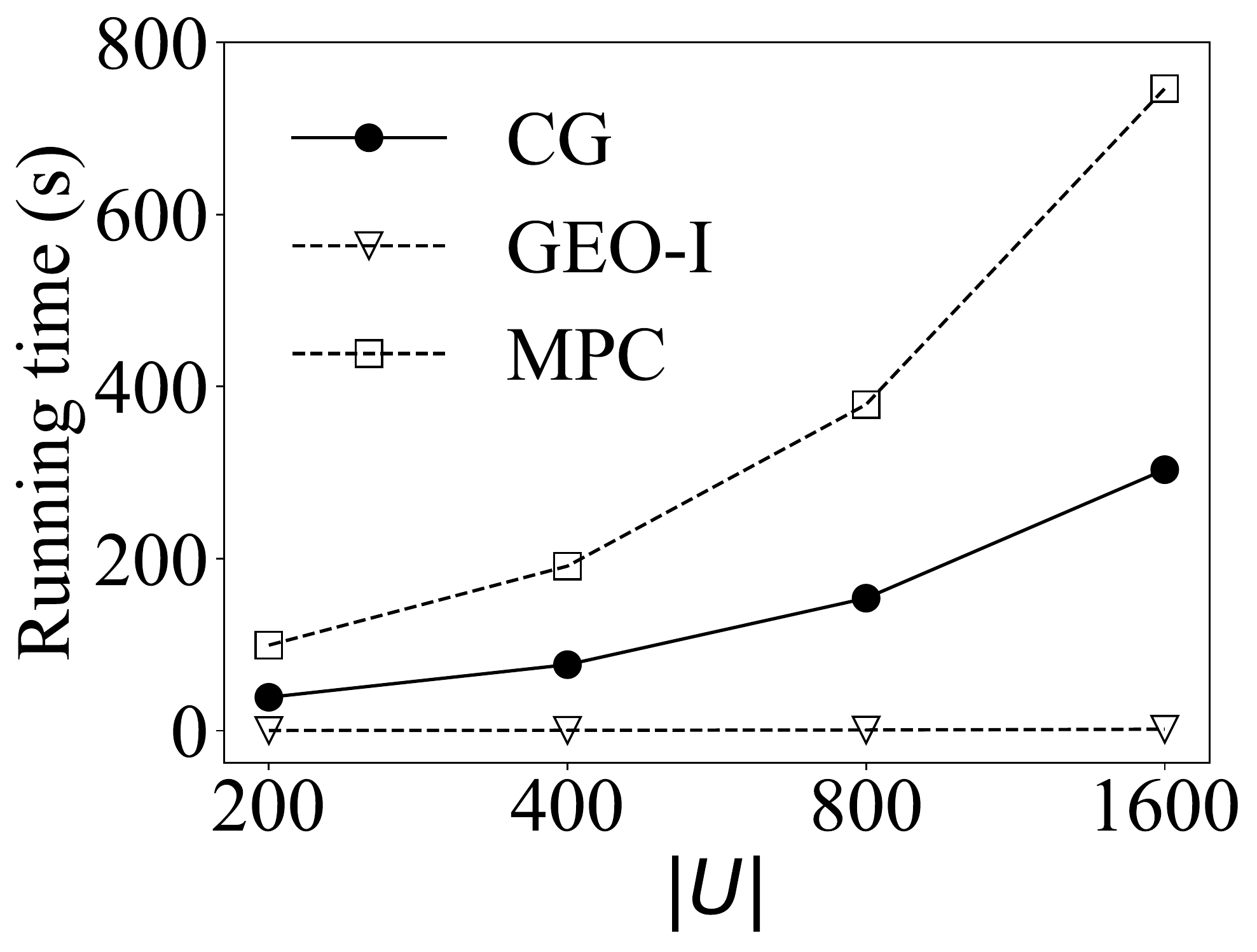}
        \caption{$\epsilon=2.0$.}
        \label{subfig:exp_time_U_eps_2}
        \end{subfigure}
       \hfill
       \begin{subfigure}[b]{0.23\textwidth}
        \centering
        \includegraphics[width=\textwidth]{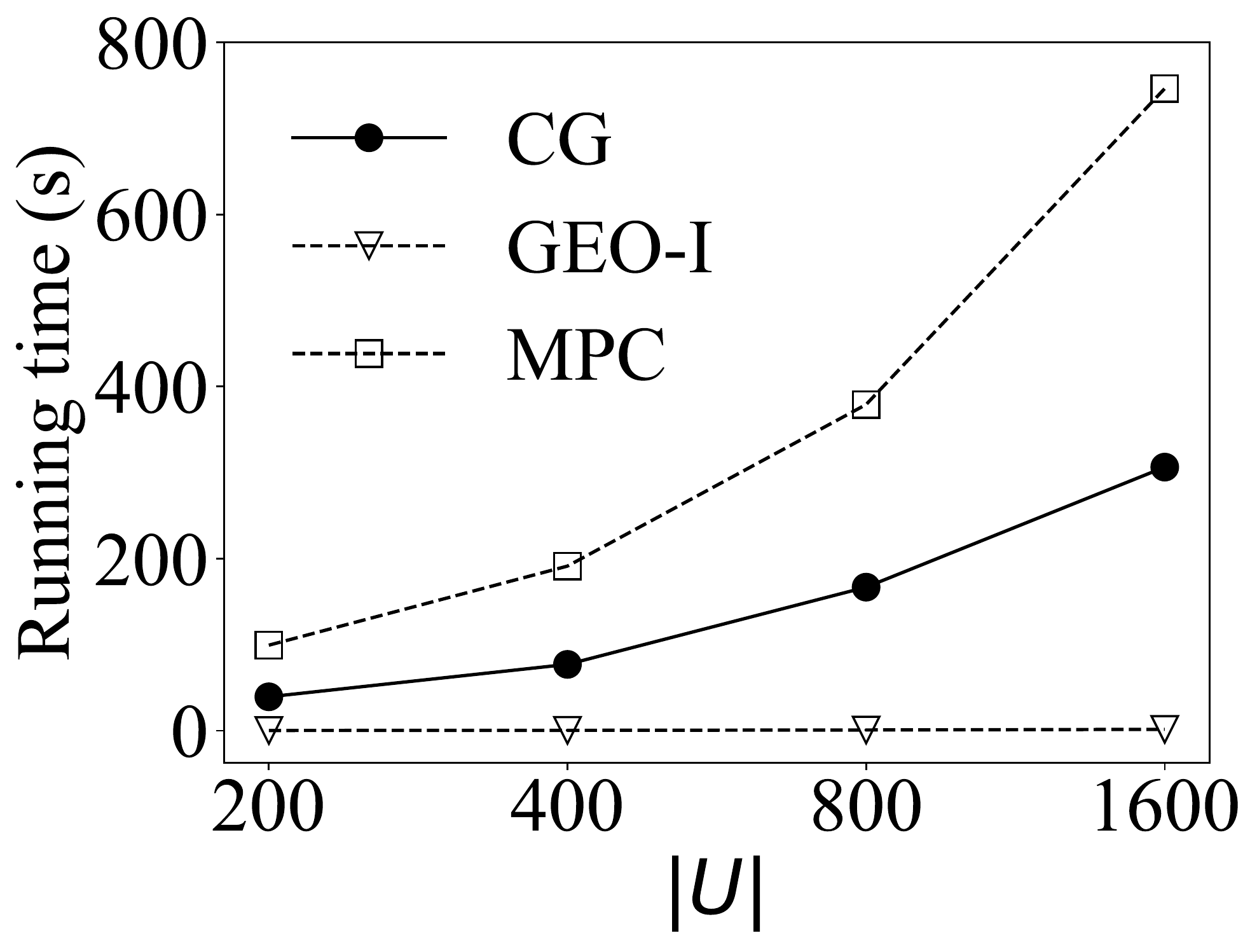}
        \caption{$\epsilon=3.0$.}
        \label{subfig:exp_time_U_eps_3}
    \end{subfigure}
    
    \caption{\small Efficiency (running time in seconds) when varying the number of users ($|U|$) on the synthetic dataset. }\label{fig:exp_time_U}
   \end{figure}

\section{Related Work}
\label{sec:relatedWork}

Our work is closely related to REACT \cite{DBLP:conf/icde/DaA0S21}, as it also uses Geo-I \cite{andres13} to preserve the location privacy. In REACT, each location is applied with Geo-I with the same privacy budget. As shown in the experiments (both in \cite{DBLP:conf/icde/DaA0S21} and in our results), the Geo-I baseline performs poorly in terms of both recall and precision, because it only uses perturbed locations. As a similar notion to Geo-I, local differential privacy has also been used to protect the trajectory data \cite{DBLP:journals/pvldb/CunninghamCFS21}. However, its main use is for aggregated statistics over relatively larger regions, where in our application, each user's contact tracing and fine-grained visited location is important. 

Many other privacy-preserving techniques are used for contact-tracing, including blockchain based P$^2$B-Trace \cite{DBLP:conf/sigmod/PengXWHXC21} and cryptographic primitives enabled BeeTrace \cite{DBLP:journals/debu/LiuTKS20} and Epione \cite{DBLP:journals/debu/TrieuSSSS20}.
As opposed to using relative locations in these works, our setting uses absolute locations, which enables tracing indirect contacts (transmission via a contaminated environment) in addition to direct contacts.

On topics not related to contact tracing, there are many recent works on combining differential privacy and cryptographic techniques \cite{wagh2021dp}. 
Secure shuffling and encryption-based techniques %\cite{bittau2017prochlo} 
\cite{erlingsson2019amplification,roy2020crypte} could help reduce the error introduced by differential privacy, and differential privacy could enable more efficient secure computation.
Our solution, ContactGuard, is inspired by this line of ideas, and uses Geo-I to accelerate MPC computations. In addition, privacy issues are raised in many Location-based services (LBS) applications \cite{DBLP:conf/icde/TaoTZSC020,tong2019two,she2017feedback,gao2016top}.

\section{Conclusion}
\label{sec:conclusion}

In this work, we study the Privacy-Preserving Contact Tracing (PPCT) problem, to identify the close-contacts of the patients while preserving the location privacy of the patients and the users. To address this problem, we propose an accurate and efficient solution named ContactGuard, which accelerates the Secure Multiparty Computation (MPC) with the help of Geo-Indistinguishability (Geo-I). Experimental results demonstrate that ContactGuard provides significant speed up over the MPC baseline, while maintains an excellent level of effectiveness (as measured by recall and precision). For future works, we look into more advanced ways to compress data \cite{DBLP:conf/sigmod/LiuS022,DBLP:conf/icde/LiuSC21} to achieve better privacy/utility tradeoff. 

\bibliography{add}
\newpage
\appendix

\section{Additional technical details}
\label{sec:appendix}

\begin{algorithm}[t]
	\DontPrintSemicolon
	\KwIn{$\epsilon, l=(x,y)$. }
	\KwOut{$l'$.}

    draw $\theta$ uniformly in $[0, 2\pi)$\;
    draw $p$ uniformly in $[0, 1)$ \;
    $d=-\frac{1}{\epsilon}(W_{-1}(\frac{p-1}{e})+1)$ \label{line:random_d}\;
    $x' = l.x + \cos(\theta)\cdot d$ \;
    $y' = l.y + \sin(\theta)\cdot d$ \;
    
	\Return{$l'=(x', y')$}\;
	\caption{\texttt{Geo-I}\cite{andres13} }
	\label{algo:geo_i}
\end{algorithm}

For easy reference, the major notations used in this paper are summarized in Table \ref{tab:table1}. 

\begin{table}
	\centering \vspace{0ex}
	{\small\scriptsize
		\caption{\small Major notations used in this paper. \label{tab:table1}}
		\begin{tabular}{l|m{6cm}}
			\hline
			{\bf Symbol} & {\bf \qquad \qquad \qquad\qquad\qquad Description} \\ \hline 
			$u, U$   & A user and the set of  $n$ users\\
			\hline
			$p, P$   & A patient and the set of $m$ patients \\
			\hline
			$l, l'$  & A location and its perturbed location\\
			\hline
			$L_u, L'_u$  & The trajectories of user $u$ and its perturbed locations\\
			\hline
			$L_P $   & The set of trajectories for all the patients  \\
			\hline
			$I, \tilde{I} $   & The index of high risk locations and the noisy indexes \\
			\hline
			$r$  & Distance threshold for determining a close contact\\
			\hline
			$r'$  & Our system parameter for determining a high-risk location\\
			\hline
			$\epsilon, \epsilon_P$  & The privacy budget for the user and the patients\\
			\hline

		\end{tabular}
	}\vspace{0ex}
\end{table}

\subsection{Application Settings and Adversary Model}

\subsubsection{Application settings}

We model the contact tracing application as a client-server model (illustrated briefly in \figref{subfig:problem}):

\fakeparagraph{Client side} Each user $u \in U$ corresponds to a client. Each user stores his/her own locations $L_u$ locally on a mobile device. 

\fakeparagraph{Server side} All trajectories by the patients (represented as $L_P$ as in \defref{def:patient}) are aggregated and stored at the server. In our contact tracing application, the server is owned by the government, or more specifically, a central public health organization (\eg the Centers of Disease Control and Prevention). The setting follows the real-world situation: governments often collect whereabouts of confirmed patients in order to minimize any further transmission, usually starting with tracing the close contacts of the patients. For example, this is enforced by legislation in Hong Kong\footnote[4]{https://www.elegislation.gov.hk/hk/cap599D!en?INDEX\_CS=N}. 

For the purpose of contact tracing, rather than storing each patient's individual trajectories, it suffices to store the union of trajectories of all patients (represented as $L_P$). This is because given a user, our goal is only to determine whether the user co-visits the nearby location as \textit{some patient} $p\in P$, without the need to identify specifically \text{which patient}. This partially enhances privacy protection, as no individual patient's trajectory is identified and stored. 

As a summary, there are two major roles in our application: the users (the clients) and the government (the server). 

For the adversary on the server, the adversary tries his/her best to obtain as much private information as possible from the users (obtaining information about the locations $L_{u_i}$). Under this model, our proposed solution needs to limit the information shared from the user to the server. When the information about $L_u$ is needed, the computation process needs to be protected with privacy guarantee, \wrt $L_u$. 

Similarly, adversary could exist on the client side. It means that the adversary tries his/her best to obtain as much private information as possible from the server (which stores the patients' locations $L_P$) or from other users (obtaining information about other users' $L_u$). Thus, our proposed solution needs to limit the information shared from the server to each user as well. When the information about $L_P$ is needed, the computation process needs to be protected with privacy guarantee, \wrt $L_P$. 

\subsubsection{Performance metrics}

Considering our specific application -- contact tracing, we measure the quality of a solution based on the following criteria: 

\fakeparagraph{(1) Effectiveness } We use recall, precision and accuracy as the key metrics for measuring the effectiveness (utility) of a potential solution. In the contact tracing applications, \textit{recall} is defined as the ratio of the correctly identified contacts among all contacts. \textit{Precision} is defined as the ratio of correctly identified contacts among all contacts identified by the solution. The \textit{accuracy} is defined as the ratio of correctly classified instances (both the correctly identified contact and the non-contacts) over all users. 

While all three metrics are crucial, recall is considered  relatively more important in the contact tracing application, as it measures the capability of a solution for eliminating false negatives (identifying as many true contacts as possible). This is crucial in applications of epidemic control, where missing a true contacts may lead to further uncontrolled transmission of the disease. Taking the epidemic control of COVID-19 in China as an example, when several cases were found in a district, all citizens of the entire city need to take mandated testing in order to minimize the chances of missing any positive cases \cite{reuter2021chinese}. 

\fakeparagraph{(2) Efficiency } We use the end-to-end running time as the metric to measure the efficiency of a proposed solution. The end-to-end running time is the time needed to process all users in $U$ to determine whether they are contacts or not. As the contact tracing application is normally executed on a daily basis, if a solution could not process all users in 24 hours, it is considered impractical in terms of efficiency. 

\subsection{Baselines}

\tabref{tab:comparison_methods2} lists the comparison of the two baselines.

\begin{table}
	\centering \vspace{0ex}
	{\small\scriptsize
		\caption{\small Comparison of baselines and ContactGuard} \label{tab:comparison_methods2}
		\begin{tabular}{l|c|c|c}
			\hline
			{\bf Methods} & {\bf MPC} & {\bf Geo-I} & {\bf ContactGuard} \\ 
			\hline 
			Accuracy & $\surd$ & $\times$ & $\surd$\\
			\hline
			Efficiency & $\times$ & $\surd$ & $\surd$\\
			\hline
		\end{tabular}
	}\vspace{0ex}
\end{table}

\begin{algorithm}[t]
	\DontPrintSemicolon
	\KwIn{$r, \delta$, \texttt{ProtocolIO} $io$. }
	\KwOut{True/False.}

    $n_1 := $ ocBroadcast($io$.n, party=1) \;
		$n_2 := $ ocBroadcast($io$.n, party=2) \label{line:broad_n2}\;
    obliv float ox\_array1 = ocToOblivFloatArray(io.x, party=1) \label{line:obv_start}\;
		obliv float ox\_array2 = ocToOblivFloatArray(io.x, party=2) \;
		obliv float oy\_array1 = ocToOblivFloatArray(io.y, party=1) \;
		obliv float oy\_array2 = ocToOblivFloatArray(io.y, party=2) \;
		obliv int time\_array1 = ocToOblivIntArray(io.time, party=1) \;
		obliv int time\_array2 = ocToOblivIntArray(io.time, party=2) \label{line:obv_end}\;
	\ForEach{$i \in [0\ldots n_1-1]$} {\label{line:obv_iter_start}
		\ForEach{$j \in [0\ldots n_2-1]$}{
				obliv float d\_x $:=$ ox\_array1[i] - ox\_array2[j]\;
				obliv float d\_y $:=$ oy\_array1[i] - oy\_array2[j]\;
				obliv float d\_square $:=$ d\_x * d\_x + d\_y*d\_y\;
				obliv int $t_1 := $time\_array1[i] \;
				obliv int $t_2 := $time\_array2[j] \;
				obliv int diff\_time $:=$ time\_difference($t_1, t_2$) \;
				obliv \If{(d\_square $<= r*r$ AND diff\_time $<= \delta$)}{ \label{line:obv_check} 
					\Return{True} \;				
				}

		}
	} \label{line:obv_iter_end}
	\Return{False}\;
	\caption{\texttt{MPC\_Compare\_Location} }
	\label{algo:mpc_baseline}
\end{algorithm}

 \algoref{algo:mpc_baseline} gives the detailed steps for the MPC baseline. We only show the main procedure that is shared and called by both the server (the patients) and the client (the user). The inputs to the procedure includes an object $io$ of type \texttt{ProtocolIO}. The server side's input object $io$ contains the trajectories $L_P$ for all the patients. The client side's input object $io$ holds the trajectory $L_u$ for user $u$. \algoref{algo:mpc_baseline} uses secure computation over the private input $io$ to determine whether the user $u$ (the client side) is a contact or not. 

From \lineref{line:obv_start} to \lineref{line:obv_end}, the routine \texttt{ocToOblivFloatArray} in Obliv-C is called to transform the private inputs (an array with \texttt{float} type) to oblivious data types to be used during the secure computation. Similarly, the routine \texttt{ocToOblivIntArray} is to transform an integer arrays to an oblivious data type. In our implementation, we store the x and y coordinates of all visited locations in float arrays (\texttt{io.x} and \texttt{io.y}). We store the timestamps as integer arrays. 

Note that when we call the routine \texttt{ocToOblivFloatArray}, an argument \texttt{party=1} is given to obtain the private inputs from the patients (the server, which we refer to as the 1st party). When \texttt{party=2} is given, the routine obtains the private inputs from the user (which we refer to as the 2nd party). After the private inputs are transformed to the oblivious data types, all computations based on these oblivious data types are guaranteed to be oblivious as well. 

The number of visited locations $|L|$ is stored in \texttt{io.n} (an integer). As we do not consider the number of visited locations as a private input, it is broadcasted and stored as $n_1$ ($|L_P|$ for the patients) and $n_2$ ($|L_u|$ for the user). 

Then we iterate over all location from the patients and all locations from the user to check whether the pair of visited locations and timestamps satisfy the spatial and temporal constraints in the contact definition in \defref{def:close_contact}. All computations are oblivious, including calculating the square of the Euclidean distance between the two locations and calculating the time difference in seconds between the two timestamps, because these computations are based on oblivious inputs (\eg ox\_array1 which is transformed from the private inputs), and the operations only involve oblivious summations and multiplications, which are provided by the underlying Obliv-C. 

At the last step (\lineref{line:obv_check}), the square of the Euclidean distance is compared with $r * r$, which is the square of the distance threshold $r$. In addition, the time difference in seconds is compared with $\delta$, which is the time difference threshold. $r$ and $\delta$ are inputs to our PPCT problem, as specified in \defref{def:ppct}. The comparison uses the security primitive \texttt{obliv if} provided by Obliv-C, such that the execution of the program (the program counters or the running time) does not disclose any information about the underlying private inputs (the oblivious variables). The body of \textit{obliv if} is executed regardless of the value of conditions.

\fakeparagraph{Time Complexity} The oblivious transformation from private inputs to oblivious data types from \lineref{line:obv_start} to \lineref{line:obv_end} takes $\mathcal{O}(n_1 + n_2)$ to complete. The main overhead occurs when we iterate over all locations from the patients and all locations from the user, which takes $\mathcal{O}(n_1 * n_2) \to \mathcal{O}(|L_P||L_u|)$ to complete. Overall, for each user, the MPC procedure \algoref{algo:mpc_baseline},  takes $\mathcal{O}(|L_P| + |L_u| + |L_P||L_u|) \to \mathcal{O}(|L_P||L_u|)$. For the server $S$ to finish processing all users, it takes $\mathcal{O}(\sum_u |L_u||L_P|) \to \mathcal{O}(|U||L_P|\max_u|L_u|)$. 

Note that here the time complexity analysis assumes that the security primitives (\eg \texttt{ocTo\\OblivFloat}, oblivious summation and multiplication, \texttt{oblivious if}) run in constant time ($\mathcal{O}(1)$). In practice, the security primitives run multiple orders of magnitude slower than running a non-secure statement. We stay at a high-level analysis in terms of the total number of execution of security primitives rather than evaluating the number of gates of the lower-level secure circuits that are generated by Obliv-C. In our proposed method ContactGuard, we optimize the running time by significantly reducing the number of secure primitives used. 

\fakeparagraph{Privacy Analysis} 
\begin{theorem}\label{theo:privacy_mpc} \algoref{algo:mpc_baseline}  satisfies the privacy requirements in \defref{def:ppct}. 
\end{theorem}
\begin{proof}
	The privacy is obviously guaranteed since we are using MPC (provided by Obliv-C) operations over private inputs. We show that other than the size of the visited location set, \algoref{algo:mpc_baseline} discloses nothing about $L_u$ for the user and $L_P$ for the patients. 

	About $L_u$: from \lineref{line:obv_start} to \lineref{line:obv_end}, are transformed to oblivious data types. The computations inside the iterations from \lineref{line:obv_iter_start} to \lineref{line:obv_iter_end} only use the oblivious variables (\eg ox\_array2). Thus, $L_u$ is strictly protected. Note that in our problem setting (\defref{def:ppct}), the length of the trajectory ($|L_u|$) is not considered sensitive, and in fact we broadcast $|L_u|$ at \lineref{line:broad_n2}, and it is publicly available to both parties. 

	About $L_P$: from \lineref{line:obv_start} to \lineref{line:obv_end}, are transformed to oblivious data types. The computations inside the iterations from \lineref{line:obv_iter_start} to \lineref{line:obv_iter_end} only use the oblivious variables (\eg ox\_array1 and time\_array1). Thus, $L_P$ is strictly protected.
\end{proof}

\subsection{ContactGurad}

It consists of the following steps. 

\textbf{Step 1. Location perturbation}: each user perturbs his/her visited location set $L_u$ to a perturbed visited location set $L'_u$. The user submits $L'_u$ to the server. 

\textbf{Step 2. Subset selection}: the server compares $L'_u$ with $L_P$, the trajectories of the patients. The server keeps track of the indexes of the high-risk locations (\eg the 3rd/4th location) that are within distance $r'$ to some visited location of the patients. Then, the server returns the list of indexes (denoted as $\tilde{I}$, with privacy protection) of the identified high-risk locations to the user. 

\textbf{Step 3. Accelerated MPC}: the user selects the subset of visited locations based on $\tilde{I}$, which is returned by the server in Step 2, and then uses MPC protocol to compare the subset of locations with the patients' trajectories $L_P$ on the server. If there is a location in the subset (together with the visited timestamp) that meets the definition of contact as \defref{def:close_contact}, the user $u$ will be identified as a contact. Otherwise, the user will not be identified as a contact. 

\subsubsection{Proofs}

\equref{eq:proof_1} is by the definition of the combined mechanism $K_{1,2}(\mathbf{l})=(K_1(l_1), K_2(l_2))$, which applies mechanism $K_1$ to the first location $l_1$ and $K_2$ to the second location $l_2$ independently. So, the probability of generating $o_1$ and $o_2$ are the products of probability of generating $o_1$ by $K_1$ and $o_2$ by $K_2$. 

\equref{eq:proof_3} is because $K_1$ satisfies $\epsilon_1$-Geo-I (\cite{andres13}) and $K_2$ satisfies $\epsilon_2$-Geo-I. By the definition of Geo-I in \cite{andres13}, the probability of generating the same output location, $o_1$, from two input locations $l_1$ and $l'_1$ are bounded by $e^{\epsilon_1\cdot d(l_1, l'_1)}$. The multiplicative distance between two distributions, $d_{\rho}(M(x),M(x'))$ in \cite{andres13} is defined as $|\ln \frac{M(x)}{M(x')}|$ (see the original paper \cite{andres13}), here we transform the $\ln$ function back to the exponential form. 

\equref{eq:proof_4} is by the definition of $d_{\infty}(\mathbf{l}$, $d_{\infty}(\mathbf{l},\mathbf{l}') = max_id(l_i, l'_i)$. Both $d(l_1, l'_1) \leq max_id(l_i, l'_i)$ and $d(l_2, l'_2) \leq max_id(l_i, l'_i)$ hold. 

\fakeparagraph{Privacy analysis} The privacy guarantee of this step is provided with the following theorem. 

\begin{proof}
Without loss of generality, let us focus on one specific location $l$ in $L_P$. Assume that $l$ is within distance $r'$ to the $i$-th location in $L'_u$. Then the exact index constructed is $\mathbf{i} = \{0, \ldots, 1, \ldots, 0\}$, where the $i$-th location is with value 1 and all the other positions are 0. Then, with the randomized response mechanism, for a given output $\mathbf{o}=\{0, \ldots, 1, \ldots, 0\}$, the probability of generating this output at the $i$-th position is $e^{\epsilon_P}/(e^{\epsilon_P}+1)$ (ignoring the probability of other positions as they are the same for the neighboring inputs $L'_P$, and will be cancelled out). 

For the neighboring inputs $L'_P$, where all positions are the same except for $l' \in L'_P$, and $l'$ be outside distance $r'$ of the $i$-th location in $L'_u$ and all the other locations in $L'_u$. The exact index constructed for this neighboring input is $\mathbf{i'} = \{0, \ldots, 0, \ldots, 0\}$, where the $i$-th position is with value 0. Then the probability of generating the same output $\mathbf{0}=\{0, \ldots, 1, \ldots, 0\}$ at the $i$-th position is $1/(e^{\epsilon_P}+1)$ (again ignoring the probabilities for the other positions). So, the ratio of the two probabilities are:
\begin{equation}
	\frac{\Pr[R(\mathbf{i})=\mathbf{o}|l\in L_P)]}{\Pr[R(\mathbf{i}')=\mathbf{o}|l'\in L'_P)]} = \frac{e^{\epsilon_P}/(e^{\epsilon_P}+1)}{1/(e^{\epsilon_P}+1)} = e^{\epsilon_P}.
\end{equation}

Thus, the subset selection step satisfies ${\epsilon_P}$-Local Differential Privacy. 
\end{proof}

\subsubsection{Algorithms}

We use a running example next to better illustrate the basic idea of subset selection. 

\begin{example}

The example is shown in \figref{subfig:subset}. In this example, the user $u$ has visited location $l_1$ at time $t_1$ and $l_2$ at time $t_2$. At Step 1. location perturbation, the user perturbs location $l_1$ to $l'_1$ and $l_2$ to $l'_1$. The set $L'_u$ comprises these two perturbed locations and is sent to the server. 

In this example, the patient visited location $l_3$ at time $t_3$. Then, the server compares $L'_u$, which contains $l'_1$ and $l'_2$, with $l_3$. Because $l'_2$ locates close enough to $l_3$ (within a distance $r'$), $l'_2$ is identified as a high-risk location (and chosen to be included in the subset), whereas $l'_1$ is not chosen. Then the server returns the indexes of the subset of high-risk locations, denoted by $I=\{2\}$ back to the user, indicating that the 2nd location is a high-risk location. 

For privacy reason, before the server returns the \textit{exact} high-risk indexes, it uses randomized response to perturb them. Then, the noisy indexes are returned to the user. In this example, $I=\{2\}$ could also be represented as $\mathbf{i} = \{0, 1\}$ indicating the 1st location is not a high-risk location (indicated by a value of 0) and the 2nd location is high-risk (indicated by a value of 1). Each bit is perturbed with randomized response. Given a privacy budget $\epsilon_P$, the original value is returned with a probability $e^{\epsilon_P}/(e^{\epsilon_P}+1)$. Otherwise, it returns a flipped value (0 to 1, 1 to 0) with probability $1/(e^{\epsilon_P}+1)$. 

\end{example}

\algoref{algo:subset_selection} describes the detailed steps to select the subset of high-risk locations. It initializes the returned index set as the empty set, and then for each perturbed location $l'$, if it is within distance of $r'$ to some location $l$ visited by the patients (stored in $L_P$). For each location, we use randomized response to transform the original value (variable \texttt{flag}) to a noisy one. In the end, the noisy index $\tilde{I}$ is returned to the user $u$. 

Note that the parameter $r'$ is a controllable variables to balance the accuracy and efficiency trade-off of ContactGuard. If $r'$ is set as a large value, then all locations would be included in the subset of high-risk locations and later be used in the MPC computations, and thus it would lead to a high overhead. If $r'$ is set as a small value, then the perturbed location needs to be close to some patients' visited location in order to be included in the subset of high-risk locations, which would then lead to more false negatives. 

\begin{algorithm}[t]
	\DontPrintSemicolon
	\KwIn{$L'_u, L_P, r', \epsilon_P$. }
	\KwOut{$\tilde{I}$.}

	$\tilde{I}:= \{\}$ \;
	$i:= 0$ \;
	$\gamma := e^{\epsilon_P}/(e^{\epsilon_P}+1)$\;
    \ForEach{$l' \in L'_u$}{
				$i = i + 1$ \;
				flag$ := $ 0 \;
				\ForEach{$l \in L_P$}{
            \If{$d(l', l) \leq r'$}{
							flag $ = $ 1 \; 
            }
        }
				
				Sample $b := \text{Bernoulli}(\gamma)$ \tcp*{Randomized response} 
				\If{b = 1} {  
					flag = 1 - flag \tcp*{Flip the result}
				}
				\If{flag = 1} {
						$\tilde{I}$.insert($i$) \;
				}

    }
	\Return{$\tilde{I}$} \;
	\caption{\texttt{Subset\_Selection} }
	\label{algo:subset_selection}
\end{algorithm}

\algoref{algo:accelerated_mpc_input} describes the steps to prepare the private inputs using only the subset of the locations for the MPC protocol. 

\begin{algorithm}[t]
	\DontPrintSemicolon
	\KwIn{$L_u, \tilde{I}$. }
	\KwOut{\texttt{ProtocolIO} $io$.}
	$io:= $ Initialize a \texttt{ProtocolIO} object\;
	$io.n = \tilde{I}.$size() \;
	$i := 0$ \;
    \ForEach{$l_{idx} \in \tilde{I}$}{
				$io.x[i] = L_u[l_{idx}].x$ \;
				$io.y[i] = L_u[l_{idx}].y$ \;
				$io.time[i] = L_u[l_{idx}].time$ \;
				$i = i + 1$\;
    }
	\Return{$io$} \;
	\caption{\texttt{Prepare\_MPC\_Inputs} }
	\label{algo:accelerated_mpc_input}
\end{algorithm}

\subsubsection{Time complexity}

We conduct time complexity analysis for each step of ContactGuard. 

Step 1, location perturbation: the client side (the user) runs \algoref{algo:perturb_location_set}. The time complexity is $\mathcal{O}(|L_u|)$, linear to the number of visited locations given in the input $L_u$.

Step 2, subset selection: the server side conducts the subset selection. The time complexity is $\mathcal{O}(|L'_u||L_P|) \to \mathcal{O}(|L_u||L_P|)$, because it compares each location of $L'_u$ against every location of $L_P$. 

Step 3, accelerated MPC: let $|I|$ be the size of the subset selected in Step 2. Then, the oblivious transformation from private inputs to oblivious data types takes $\mathcal{O}(|n_I| + |L_P|)$ to complete. The main overhead occurs when we iterate over all locations from the patients and only locations from the subset of the user's visited locations, which takes $\mathcal{O}(|L_P||I|)$ to complete. Overall, for each user, the MPC procedure takes $\mathcal{O}(|L_P| + |I| + |L_P||I|) \to \mathcal{O}(|L_P||I|)$.  

Overall, for the server to finish processing all users, it takes $\mathcal{O}(\sum_u |I_u||L_P|) \to \mathcal{O}(|U||L_P| \\ \max_u|I_u|)$, where $I_u$ indicates the subset for each user $u$. The time complexity takes one MPC operation as a unit operation, and it overrides the $\mathcal{O}(|L_u||L_P|)$ time needed at Step 2, subset selection, which does not require secure computation. 

\subsubsection{Privacy analysis}

We show the end-to-end privacy analysis for ContactGuard. 

Step 1, location perturbation: $L_u$ is protected with $\epsilon$-General-Geo-I as shown in \theref{theo:general_geo_i}. The client side does not have knowledge about the patients' locations ($L_P$), so Step 1 is private \wrt both $L_u$ and $L_P$. 

Step 2, subset selection: the computation is only done on the perturbed location $L'_u$. According to the post-processing property of Geo-I, the computation process is still private \wrt $L_u$. About the patients' locations $L_P$, \theref{theo:privacy_subset} shows that the returned noisy index satisfies $\epsilon_P$-LDP \wrt $L_P$. Thus, this step is also private \wrt both $L_u$ and $L_P$. 

Step 3, accelerated MPC: the privacy is guaranteed by the MPC procedure (provided by Obliv-C) operations over private inputs. On the client side, the subset of $L_u$ is selected from $L_P$ with the noisy index $\tilde{I}$, and then they are transformed to the oblivious data types. On the server, $L_P$ is also transformed to the oblivious data types. All further computations are based on oblivious operators.

\section{Additional experimental results}
\label{sec:appendix_exp}

Gowalla dataset: it is a location-based social network check-in dataset \cite{DBLP:conf/kdd/ChoML11}, recording the latitude and the longitude of users' visited locations. The dataset contains records in different cities around the world, and we select the records from San Francisco city, allowing the GIS coordinates to be mapped to X,Y Cartesian coordinates in the range of $[0, 10549] \times [0, 8499]$, such that the Geo-I could be directly applied to the X,Y coordinates. Following the work \cite{DBLP:conf/kdd/DuTZTZ18}, we then randomly select a small subset of the users as the patients (the number of patient is set as 2 to 8, with the patient/user ratio set as around 1\%), and the rest are set as users to be tested. 

Synthetic dataset: we generate each location with X,Y coordinates from the range $[0, 10549] \times [0, 8499]$, which is the same as the one of the Gowalla dataset. 

\subsection{System configuration}
 The experiment is conducted on a CentOS Linux system with Intel(R) Xeon(R) Gold 6240R CPU @2.40GHz and 1007G memory. We implement the methods and conduct the experiments using C. The MPC extension is provided by Obliv-C \cite{zahur2015obliv}. We use two processes on the same machine to simulate the client-server MPC protocol, and neglect the communication time.

The detailed results are shown in \tabref{tab:acc_time_tradeoff} and \tabref{tab:acc_time_tradeoff_2nd}, respectively.

As compared to the Geo-I baseline, CG method obtains 2.28$\times$, 2.38$\times$, and 1.24$\times$ improvement in terms of recall, precision and accuracy, respectively. On the other hand, when we compare to the \textit{exact} MPC baseline, the CG method is about 2.5$\times$ faster (in both \tabref{tab:acc_time_tradeoff} and \tabref{tab:acc_time_tradeoff_2nd}).

\begin{table}[t]
	\centering
	\caption{Recall/Precision/Accuracy vs. Running time trade-off of different methods on the Gowalla dataset.  $|U|=200, \epsilon=4.0.$}
	\label{tab:acc_time_tradeoff}
	\begin{tabular}{lccc}
		\toprule
		Method & Geo-I & MPC & ContactGuard   \\
		\midrule
		Recall  & 70.0\% & 100\% & \textbf{100\%}    \\
		Precision & 46.7\% & 100\% & \textbf{100\%}   \\
        $F_1$ score & 56.0\% & 100\% & \textbf{100\%} \\
        Accuracy & 76.8\% & 100\% & \textbf{100\%}\\
        Running time (s) & 0.08 & 99.35 & \textbf{39.31 (2.53x faster)}\\
		\bottomrule
	\end{tabular}
\end{table}

\begin{table}[t]
	\centering
	\caption{Recall/Precision/Accuracy vs. Running time trade-off of different methods on the Gowalla dataset.  $|U|=400, \epsilon=3.0.$}
	\label{tab:acc_time_tradeoff_2nd}
	\begin{tabular}{lccc}
		\toprule
		Method & Geo-I & MPC & ContactGuard   \\
		\midrule
		Recall  & 39.0\% & 100\% & \textbf{88.89\%}    \\
		Precision & 42.0\% & 100\% & \textbf{100\%}   \\
        $F_1$ score & 35.0\% & 100\% & \textbf{94.1\%} \\
        Accuracy & 80.5\% & 100\% & \textbf{99.75\%}\\
        Running time (s) & 0.39 & 194.99 & \textbf{77.03 (2.52x faster)}\\
		\bottomrule
	\end{tabular}
\end{table}

For a different setting where the privacy budget is smaller ($|U|=400, \epsilon=3.0$, as in \figref{fig:exp_recall_prec_time_2nd}), CG obtains a reasonably high level of effectiveness. In addition to maintaining a 100\% precision, it obtains 88.89\% in recall and 99.75\% in accuracy. 

\begin{figure}[t]
    \centering
       \begin{subfigure}[b]{0.23\textwidth}
           \centering
           \includegraphics[width=\textwidth]{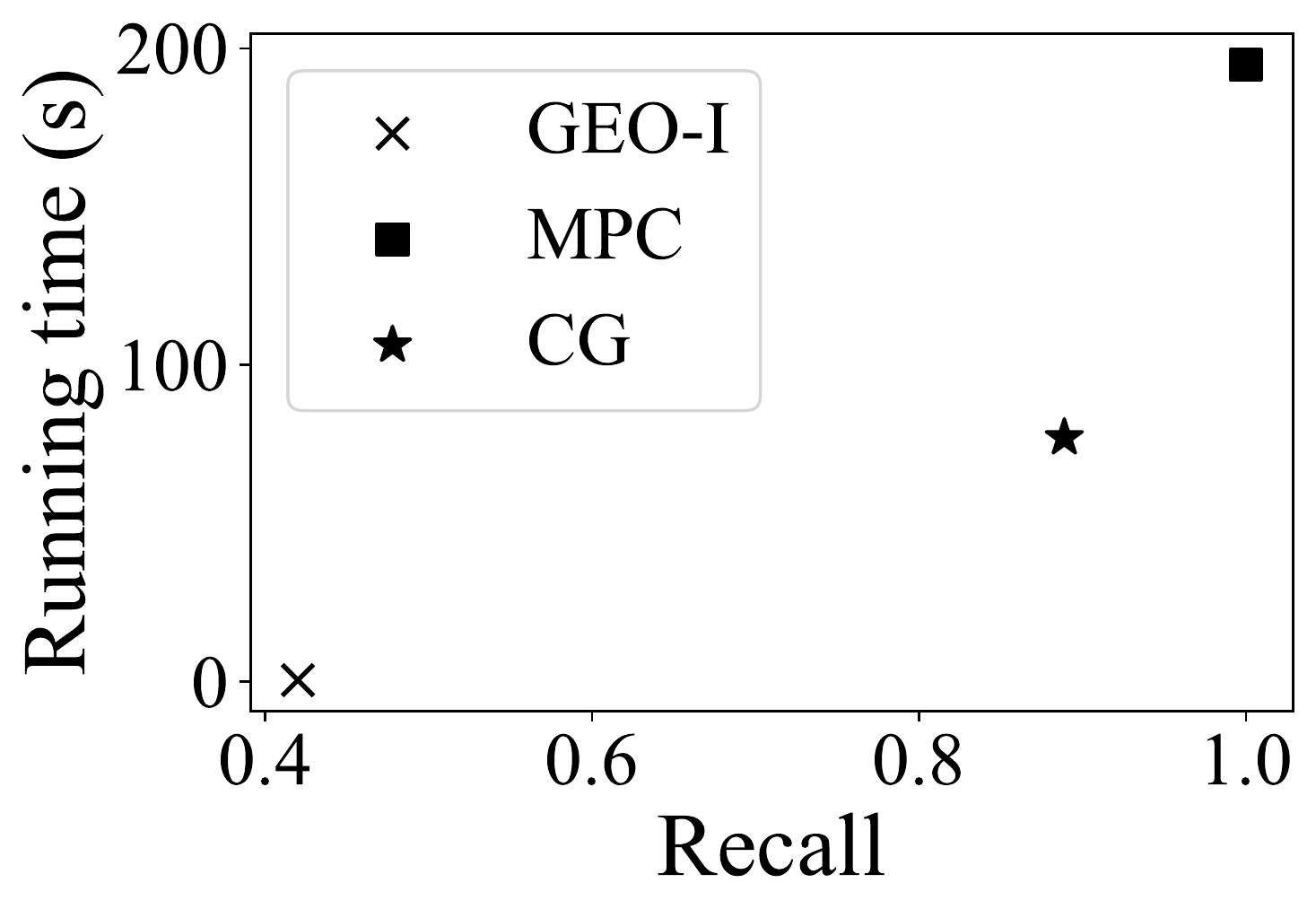}
           \caption{Recall vs. Running time.}
           \label{subfig:exp_recall_time_2nd}
       \end{subfigure}
       \hfill
        \begin{subfigure}[b]{0.23\textwidth}
           \centering
           \includegraphics[width=\textwidth]{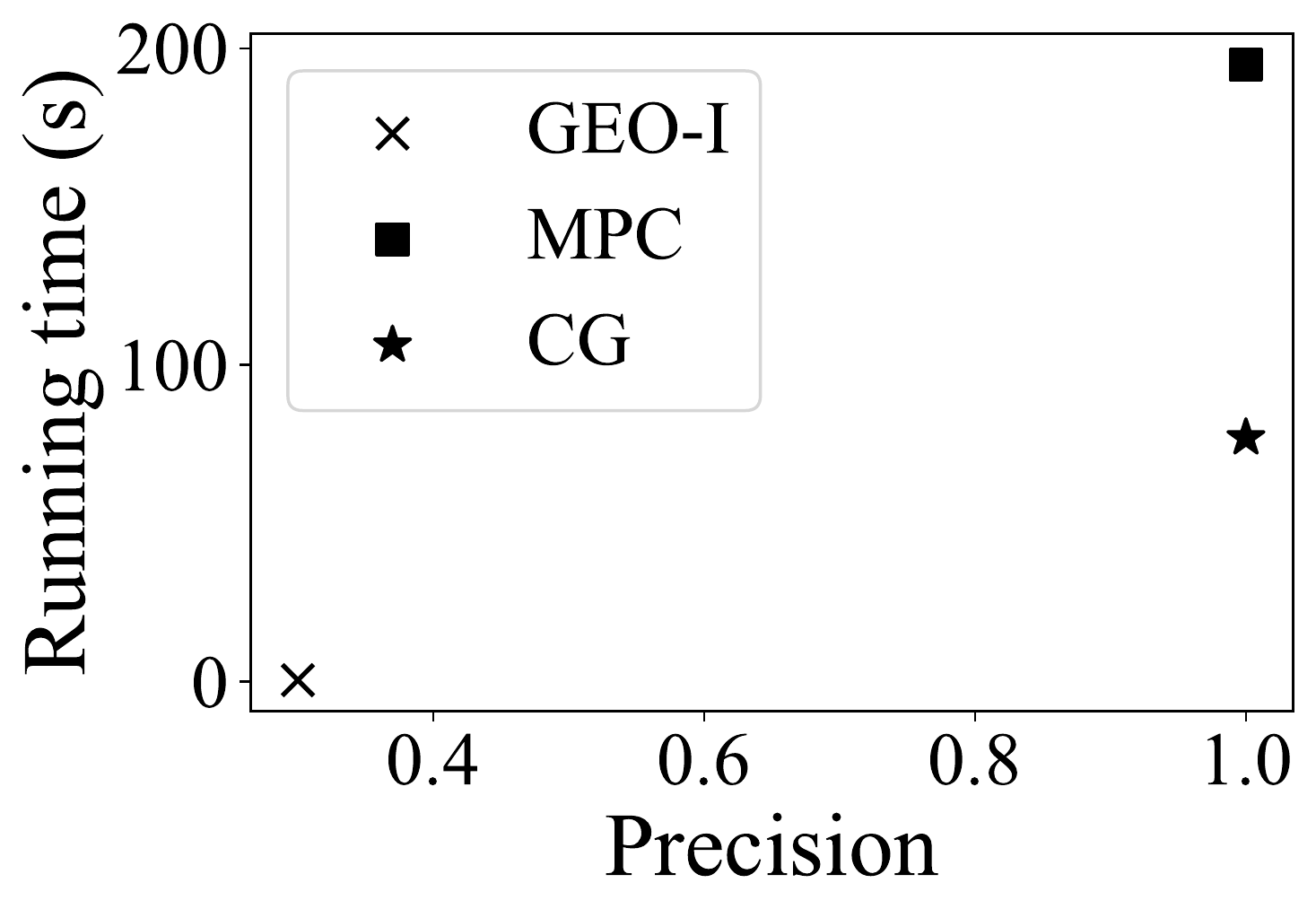}
           \caption{Precision vs. Running time.}
           \label{subfig:exp_prec_time_2nd}
       \end{subfigure}
       \hfill
        \begin{subfigure}[b]{0.23\textwidth}
           \centering
           \includegraphics[width=\textwidth]{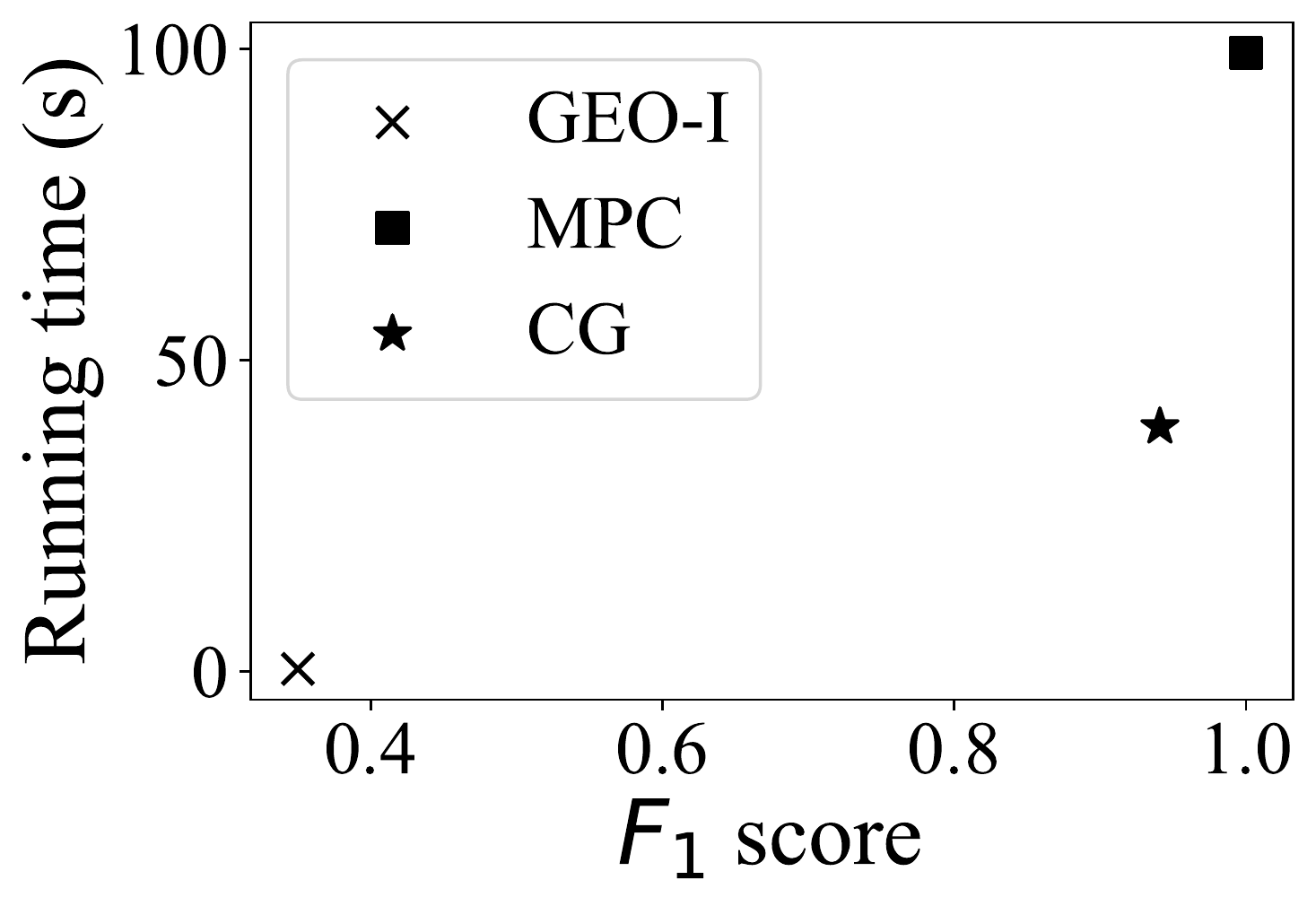}
           \caption{$F_1$ score vs. Running time.}
           \label{subfig:exp_f1_time_2nd}
       \end{subfigure}
       \hfill
        \begin{subfigure}[b]{0.23\textwidth}
           \centering
           \includegraphics[width=\textwidth]{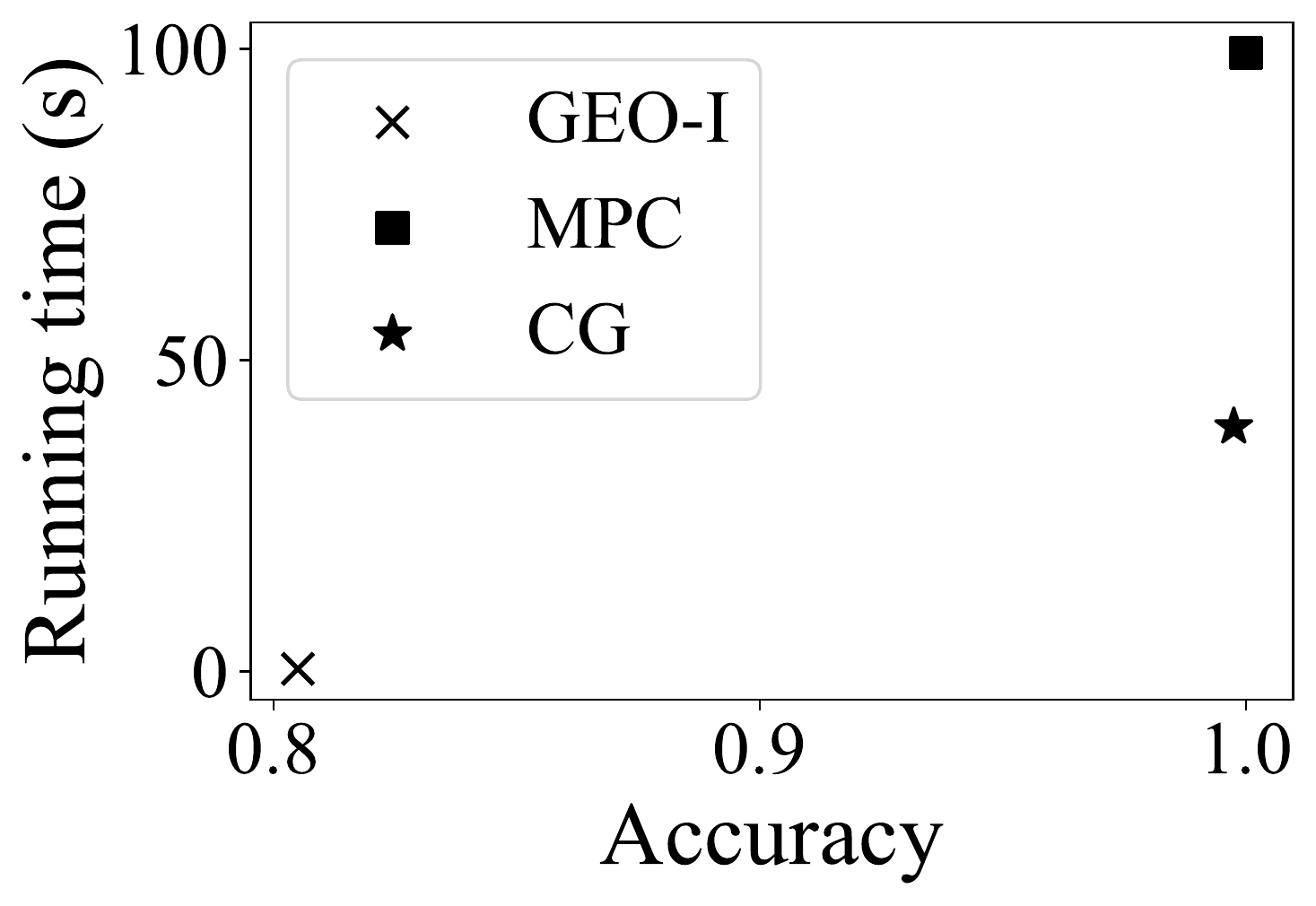}
           \caption{Accuracy vs. Running time.}
           \label{subfig:exp_acc_time_2nd}
       \end{subfigure}
    \caption{\small Effectiveness (Recall/Precision/$F_1$/Accuracy) vs. Running time trade-off of different methods on the Gowalla dataset.  $ |U|=400, \epsilon=3.0.$}\label{fig:exp_recall_prec_time_2nd}
\end{figure}

\figref{fig:syn_recall_U} shows measures of effectiveness (recall, precision, $F_1$, and accuracy) when we vary the number of users ($|U|$), for the synthetic dataset. 

\begin{figure}[t]
    \centering
       \begin{subfigure}[b]{0.23\textwidth}
           \centering
           \includegraphics[width=\textwidth]{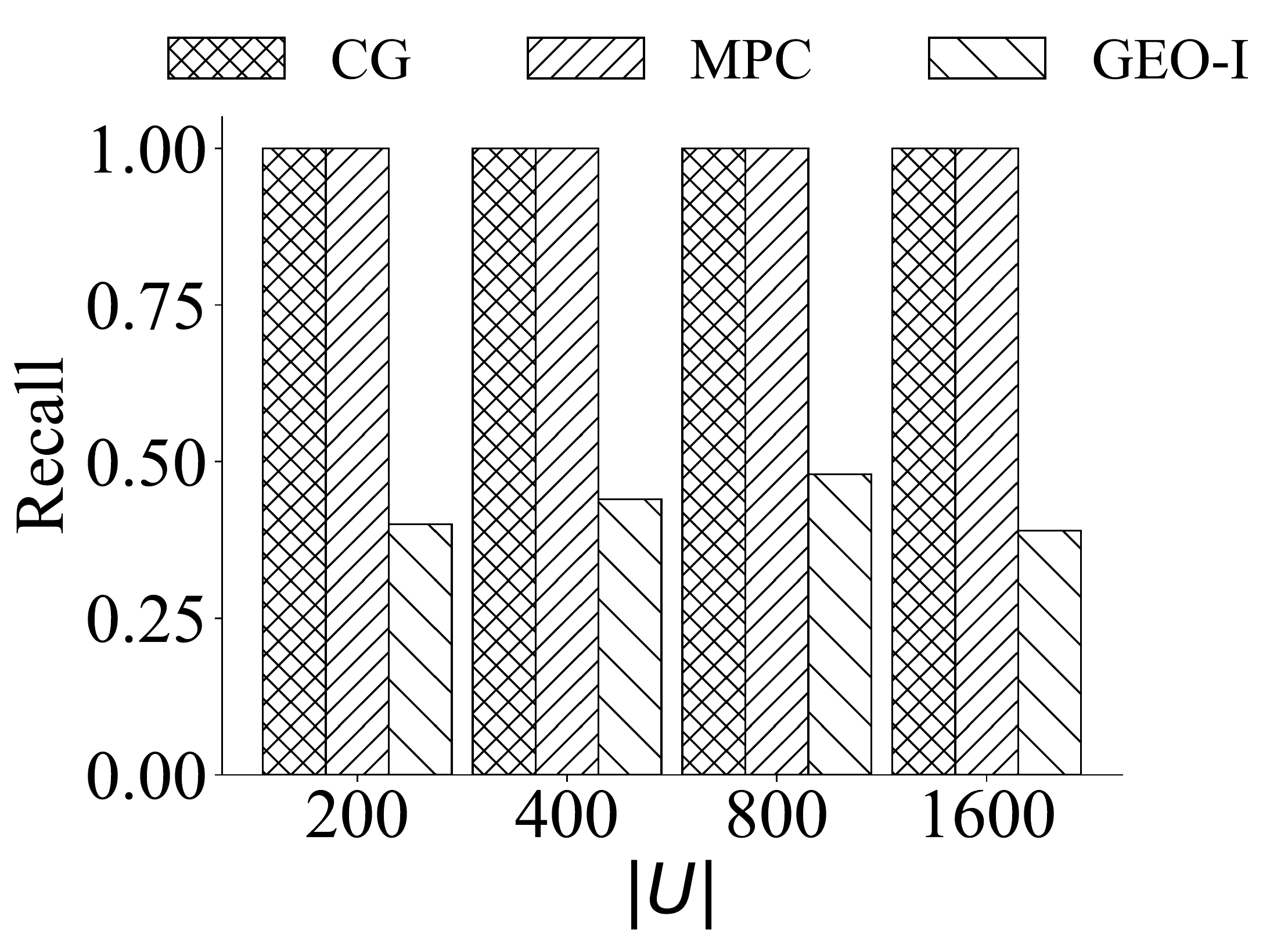}
           \caption{Recall.}
           \label{subfig:syn_recall_U}
       \end{subfigure}
       \hfill
        \begin{subfigure}[b]{0.23\textwidth}
           \centering
           \includegraphics[width=\textwidth]{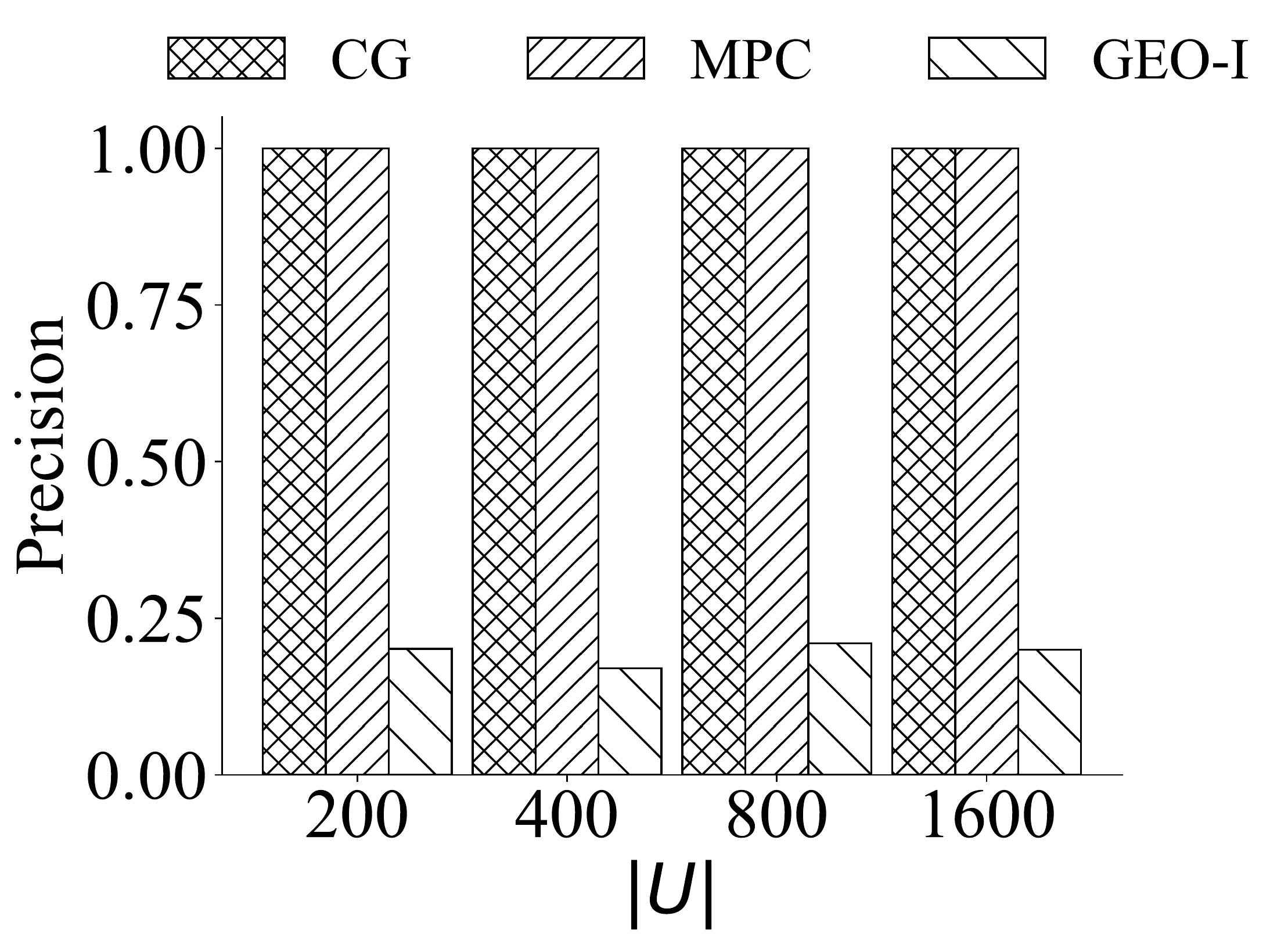}
           \caption{Precision.}
           \label{subfig:syn_prec_U}
       \end{subfigure}
       \hfill
        \begin{subfigure}[b]{0.23\textwidth}
           \centering
           \includegraphics[width=\textwidth]{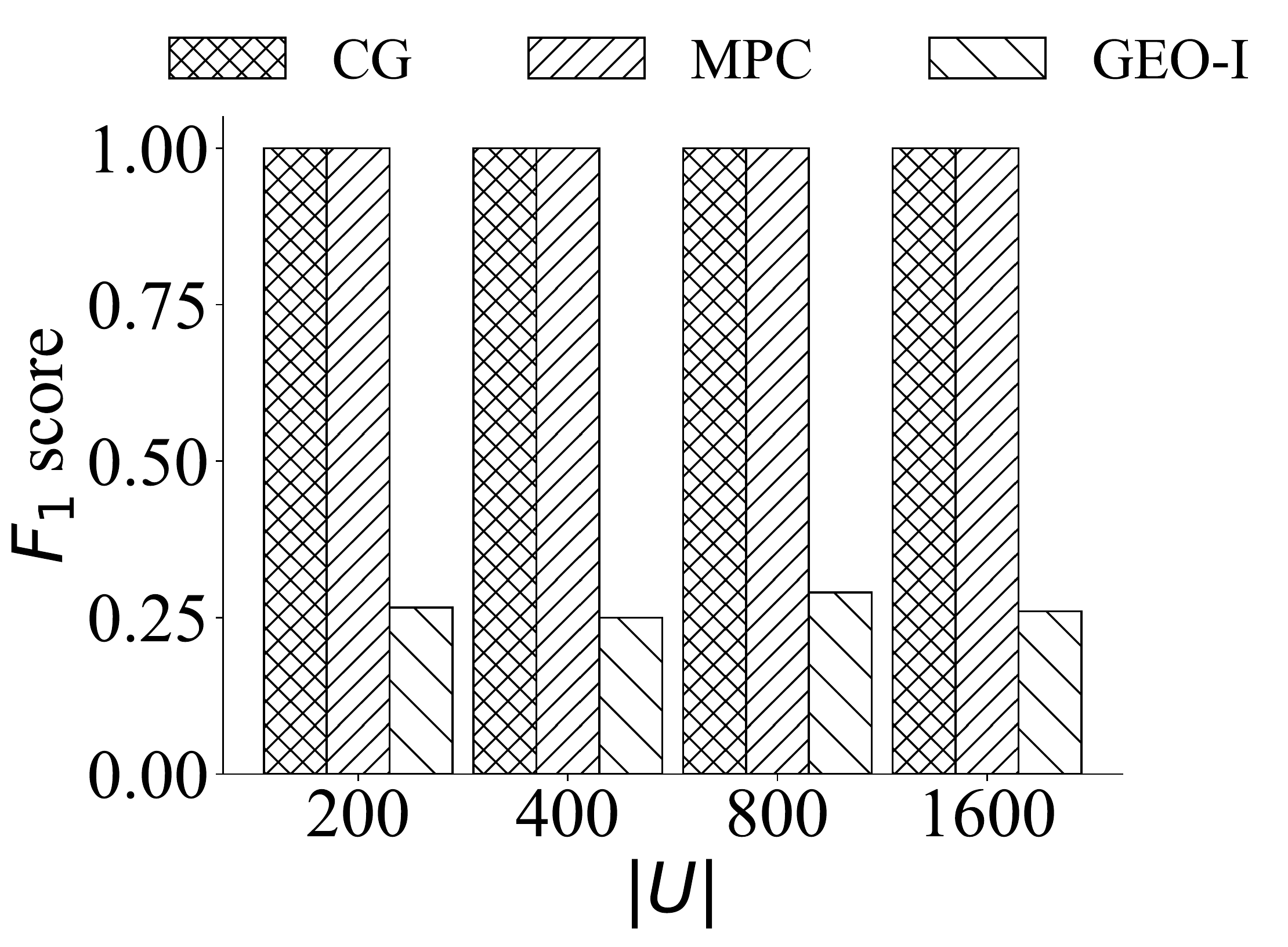}
           \caption{$F_1$ score.}
           \label{subfig:syn_f1_U}
       \end{subfigure}
       \hfill
        \begin{subfigure}[b]{0.23\textwidth}
           \centering
           \includegraphics[width=\textwidth]{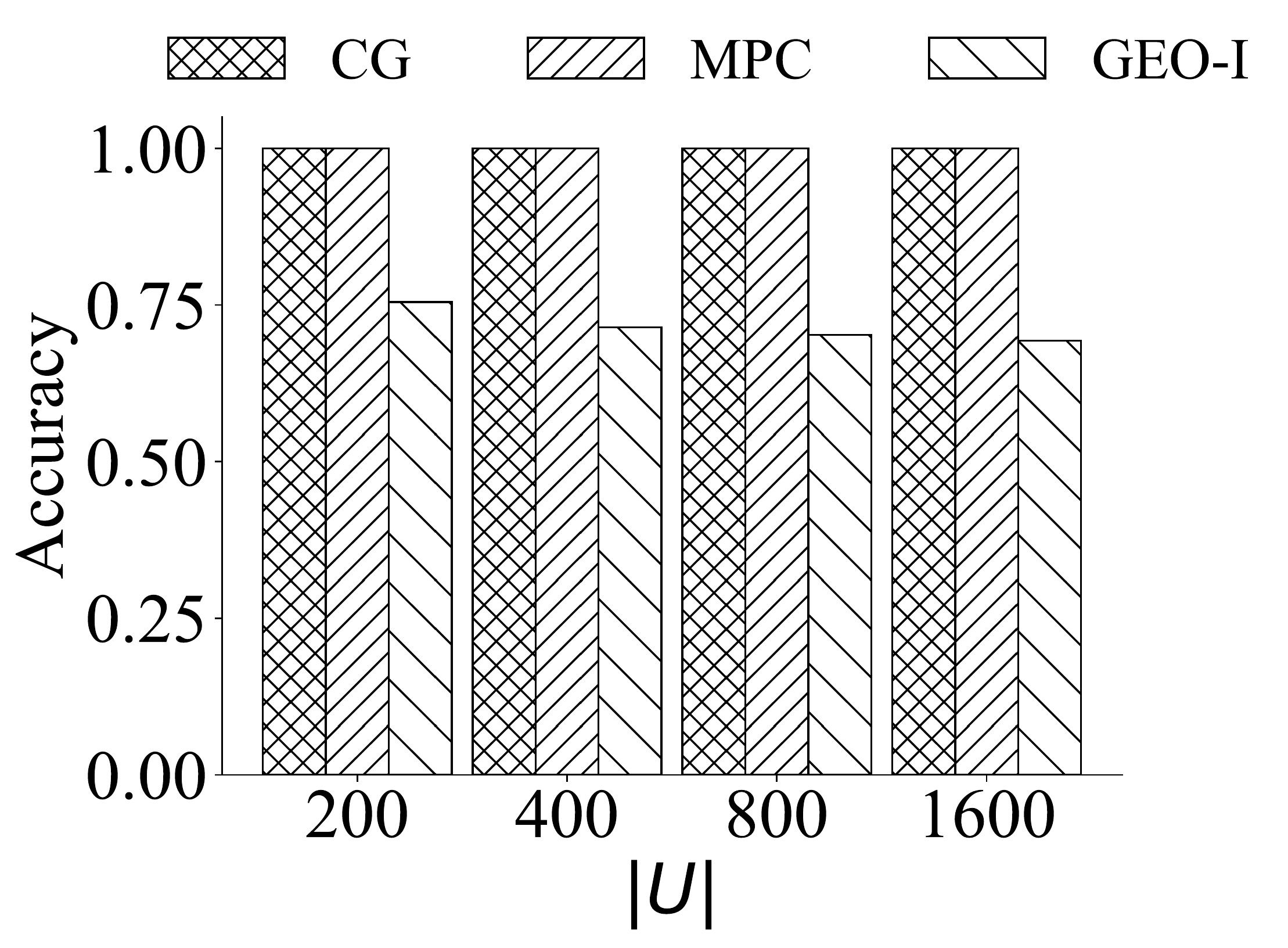}
           \caption{Accuracy.}
           \label{subfig:syn_acc_U}
       \end{subfigure}
    \caption{\small Effectiveness (Recall/Precision/$F_1$/Accuracy) when varying the number of users ($|U|$) on the synthetic dataset.  $ \epsilon=4.0.$}\label{fig:syn_recall_U}
\end{figure}

\subsection{More results on the Gowalla dataset}

\noindent \textit{More on varying $\epsilon_P$:}

The CG method obtains high effectiveness as measured by recall/\\precision/$F_1$/accuracy when we vary $\epsilon_P$. Though the recall drops to around 0.8 when $\epsilon_P=2.0$, it obtains a similar level as to the one of MPC when a larger value is given.

\begin{figure}[t]
        % \vspace{-5ex}
	\centering
		 \begin{subfigure}[b]{0.23\textwidth}
				 \centering
				 \includegraphics[width=\textwidth]{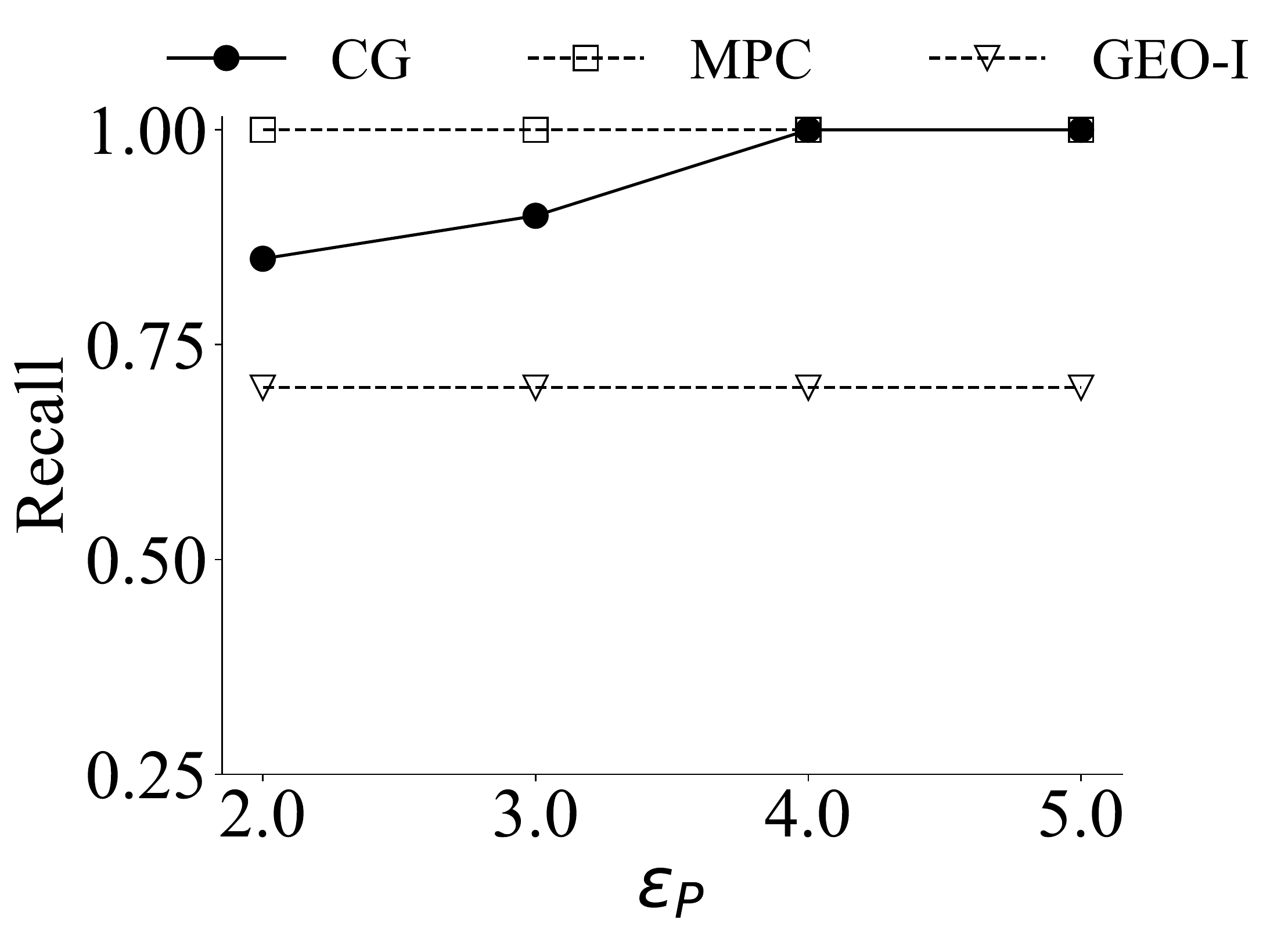}
				 \caption{Recall.}
				 \label{subfig:recall_eps_P}
		 \end{subfigure}
		 \hfill
			\begin{subfigure}[b]{0.23\textwidth}
				 \centering
				 \includegraphics[width=\textwidth]{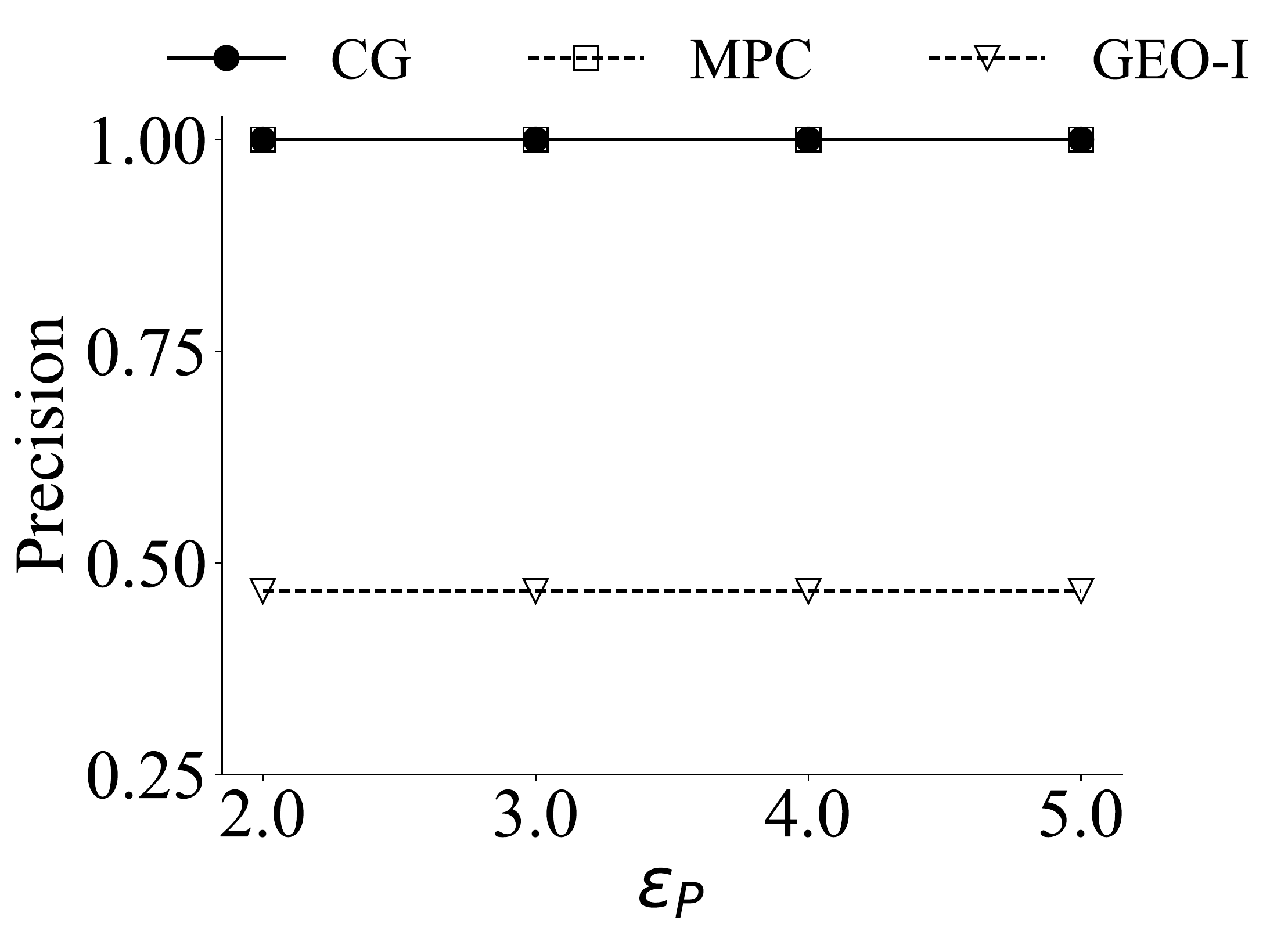}
				 \caption{Precision.}
				 \label{subfig:prec_eps_P}
		 \end{subfigure}
		 \hfill
			\begin{subfigure}[b]{0.23\textwidth}
				 \centering
				 \includegraphics[width=\textwidth]{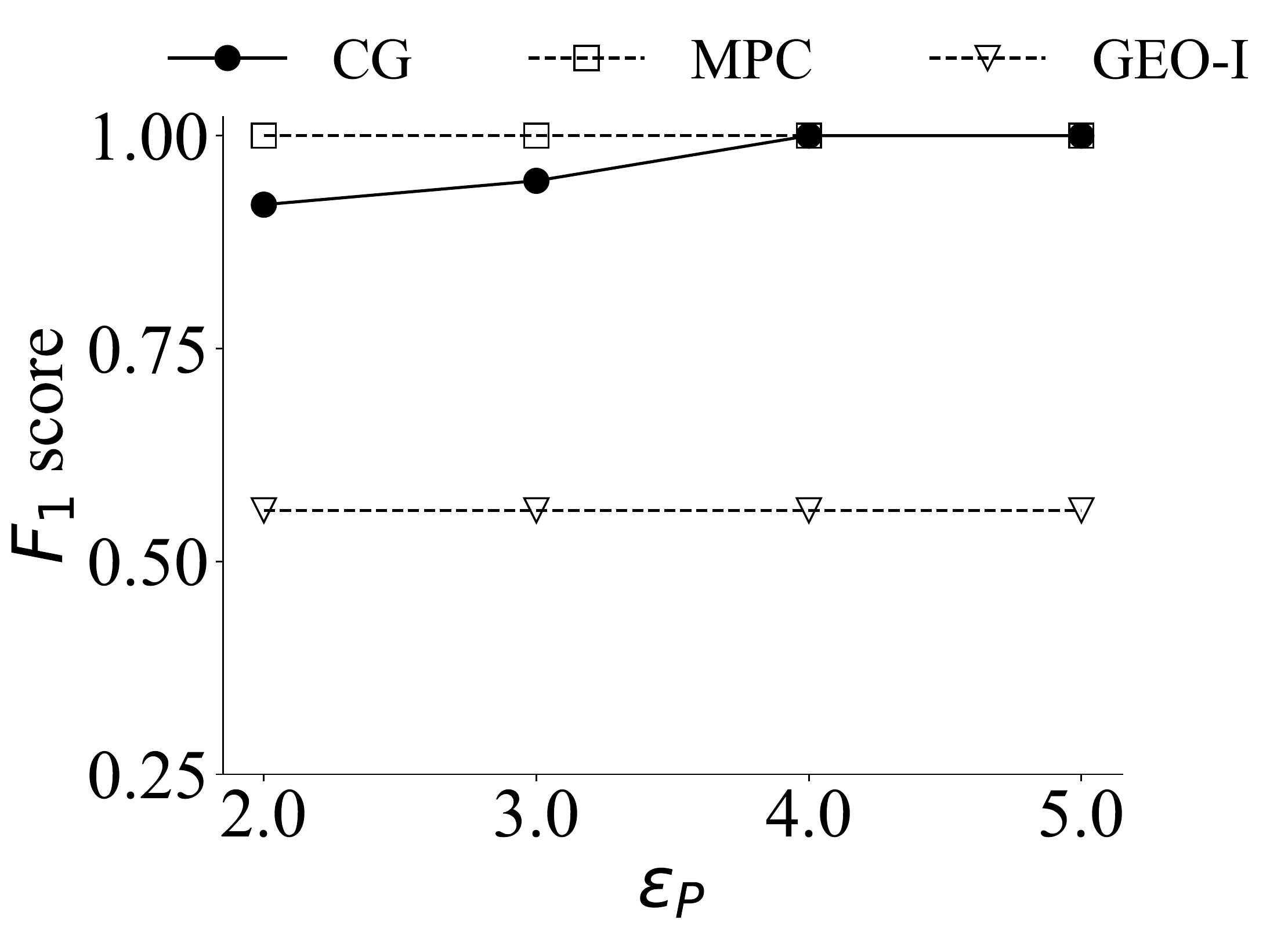}
				 \caption{$F_1$ score.}
				 \label{subfig:f1_eps_P}
		 \end{subfigure}
		 \hfill
			\begin{subfigure}[b]{0.23\textwidth}
				 \centering
				 \includegraphics[width=\textwidth]{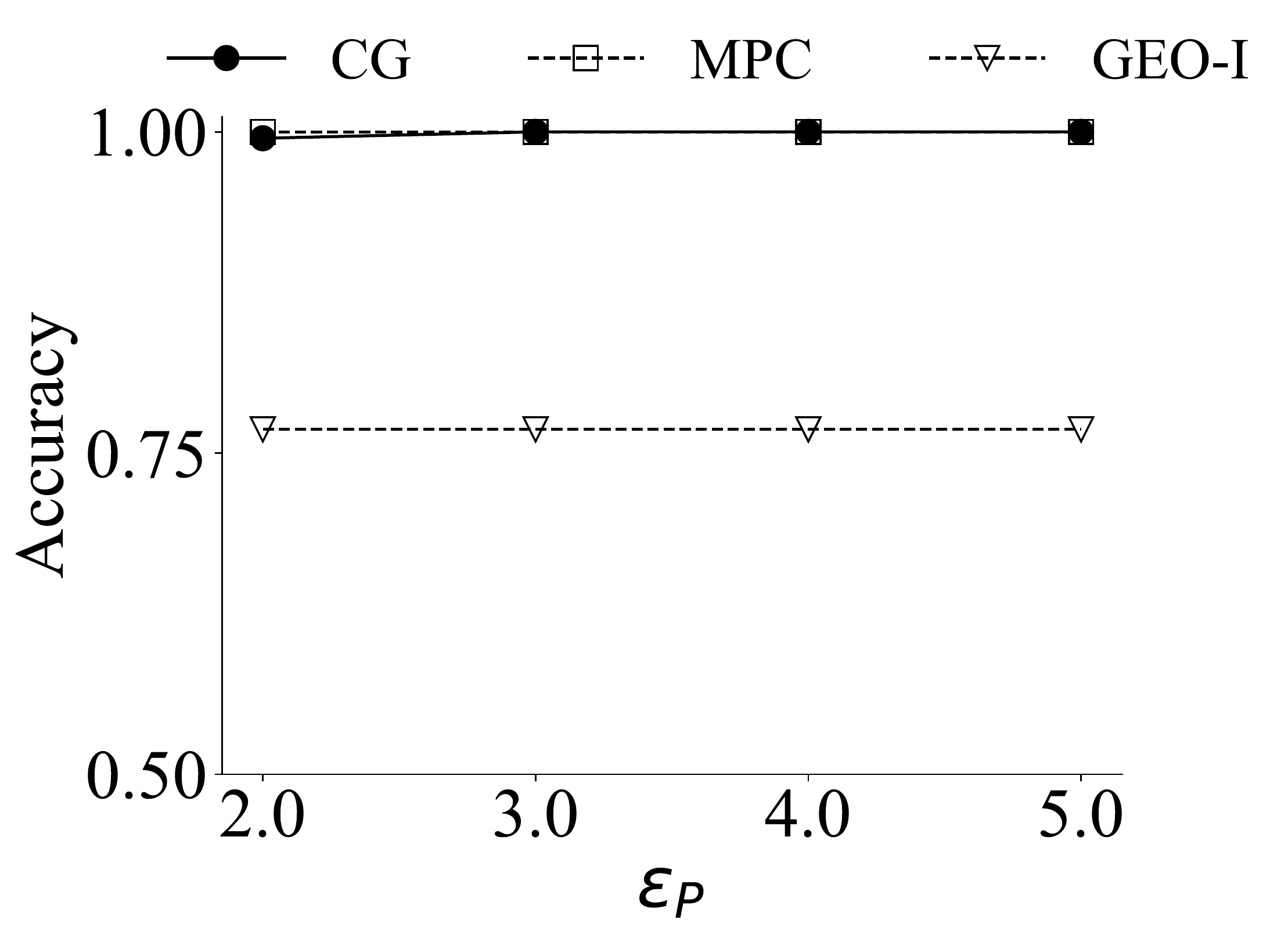}
				 \caption{Accuracy.}
				 \label{subfig:acc_eps_P}
		 \end{subfigure}
	\caption{\small Effectiveness (Recall/Precision/$F_1$/Accuracy) when varying the privacy budget for the patients ($\epsilon_P$) on the Gowalla dataset.  $eps=4.0, |U|=400.$}\label{fig:exp_effectiveness_eps_P}
\end{figure}

\bibliographystyle{plain}

\end{document}